\documentclass[a4paper,UKenglish,cleveref, autoref]{lipics-v2019}
\newif\ifPreprint \Preprintfalse
\Preprinttrue

\graphicspath{{./graphics/}}%helpful if your graphic files are in another directory

\title{Verifying that a compiler preserves concurrent value-dependent information-flow security}

\titlerunning{Verifying that a compiler preserves concurrent value-dependent infoflow security}%optional, please use if title is longer than one line

\author{Robert Sison}{Data61, CSIRO,
Australia \and UNSW Sydney, Australia
}{Robert.Sison@data61.csiro.au}{https://orcid.org/0000-0003-0313-9764}{Australian Government RTP Scholarship \& Data61 Research Project Award}% mandatory, please use full name; only 1 author per \author macro; first two parameters are mandatory, other parameters can be empty. Please provide at least the name of the affiliation and the country. The full address is optional

\author{Toby Murray}{University of Melbourne,
Australia}{toby.murray@unimelb.edu.au}{}{}%{[orcid]}{[funding]}

\authorrunning{R. Sison and T. Murray}% mandatory. First: Use abbreviated first/middle names. Second (only in severe cases): Use first author plus 'et al.'

\Copyright{Robert Sison and Toby Murray}% mandatory, please use full first names. LIPIcs license is "CC-BY";  http://creativecommons.org/licenses/by/3.0/

%TODO mandatory; please add comma-separated list of keywords

\category{}%optional, e.g. invited paper

\ifPreprint
\relatedversion{}%optional, e.g. full version hosted on arXiv, HAL, or other respository/website
\else
\relatedversion{The extended version of the paper is available at {\color{red}(FIXME: add arXiv link)}.}
\fi

\supplement{The Isabelle/HOL theories are available at \url{https://covern.org/itp19.html}.}%optional, e.g. related research data, source code, ... hosted on a repository like zenodo, figshare, GitHub, ...

\acknowledgements{We would like to thank our anonymous reviewers, as well as Carroll Morgan, Kai Engelhardt, Gerwin Klein, Christine Rizkallah, Matthew Brecknell, Johannes \r{A}man Pohjola, and Qian Ge, for their very helpful feedback on earlier versions of this paper.}%optional

\nolinenumbers %uncomment to disable line numbering

\EventEditors{John Harrison, John O'Leary, and Andrew Tolmach}
\EventNoEds{3}
\EventLongTitle{10th International Conference on Interactive Theorem Proving (ITP 2019)}
\EventShortTitle{ITP 2019}
\EventAcronym{ITP}
\EventYear{2019}
\EventDate{September 9--12, 2019}
\EventLocation{Portland, OR, USA}
\EventLogo{}
\SeriesVolume{141}
\ArticleNo{5}
\usepackage{todonotes}

\newsavebox{\mybox}

\usepackage{multido}
\newcommand{\kw}[1]{\textbf{#1}}
\newcommand{\myindent}[1]{{\multido{}{#1}{\qquad}}}
\newcommand{\Seq}[2]{#1;\\#2}
\newcommand{\Skip}{\kw{skip}}
\newcommand{\Stop}{\kw{stop}}
\newcommand{\LockAcq}[1]{\kw{lock}(#1)}
\newcommand{\LockRel}[1]{\kw{unlock}(#1)}
\newcommand{\Assign}[2]{#1\ensuremath{{}\mathbin{:=}{}}#2}
\newcommand{\ITE}[3]{\kw{if}\ #1\ \kw{then}\\#2\\\kw{else}\\#3\\\kw{fi}}
\newcommand{\ITEi}[4]{\myindent{#1}\kw{if}\ #2\ \kw{then}\\#3\\\myindent{#1}\kw{else}\\#4\\\myindent{#1}\kw{fi}}
\newcommand{\While}[2]{\kw{while}\ #1\ \kw{do}\\#2\\\kw{od}}
\newcommand{\Whilei}[3]{\myindent{#1}\kw{while}\ #2\ \kw{do}\\#3\\\myindent{#1}\kw{od}}
\newcommand{\var}[1]{\ensuremath{\mathit{#1}}}
\newcommand{\const}[1]{\textsf{#1}\xspace}
\newcommand{\ITEg}[3]{\kw{if}\ #1\ \kw{then}\ #2\ \kw{else}\ #3\ \kw{fi}}
\newcommand{\Whileg}[2]{\kw{while}\ #1\ \kw{do}\ #2\ \kw{od}}
\newcommand{\Seqg}[2]{#1\ensuremath{{}\mathbin{;}{}}#2}

\newcommand{\type}[1]{\mathit{#1}}
\newcommand{\Lock}{\type{Lock}}  % the set of locks
\newcommand{\Mode}{\type{Mode}}    % the set of modes
\newcommand{\Mem}{\type{Mem}}    % shared memories
\newcommand{\Var}{\type{Var}}    % the set of shared variable identifiers
\newcommand{\Reg}{\type{Reg}}    % the set of register identifiers
\newcommand{\cmd}{\type{cmd}}
\newcommand{\Val}{\type{Val}}
\newcommand{\expr}{\type{exp}}
\newcommand{\set}{\type{set}}
\newcommand{\Lab}{\type{Lab}}
\newcommand{\nat}{\type{nat}}
\newcommand{\oftype}{::}
\newcommand{\deftype}{::=}
\newcommand{\Load}{\kw{Load}}
\newcommand{\Store}{\kw{Store}}
\newcommand{\Jmp}{\kw{Jmp}}
\newcommand{\Jz}{\kw{Jz}}
\newcommand{\Nop}{\kw{Nop}}
\newcommand{\MoveK}{\kw{MoveK}}
\newcommand{\MoveR}{\kw{MoveR}}
\newcommand{\Op}{\kw{Op}}
\newcommand{\RISCLockAcq}{\kw{LockAcq}}
\newcommand{\RISCLockRel}{\kw{LockRel}}
\newcommand{\LocalConf}[3]{\langle #1, #2, #3 \rangle}
\newcommand{\LocalConfAbs}[3]{\langle #1, #2, #3 \rangle_\mathsf{A}}
\newcommand{\LocalConfConc}[3]{\langle #1, #2, #3 \rangle_\mathsf{C}}
\newcommand{\LocalConfWhile}[3]{\langle #1, #2, #3 \rangle_\mathsf{w}}
\newcommand{\LocalConfRISC}[5]{\langle ((#1, #2), #3), #4, #5 \rangle_\mathsf{r}}
\newcommand{\EvalStep}{\rightsquigarrow}
\newcommand{\NEvalStepAbs}[1]{\rightsquigarrow_\mathsf{A}^{#1}}
\newcommand{\EvalStepConc}{\rightsquigarrow_\mathsf{C}}
\newcommand{\NEvalStepWhile}[1]{\rightsquigarrow_\mathsf{w}^{#1}}
\newcommand{\EvalStepRISC}{\rightsquigarrow_\mathsf{r}}
\newcommand{\thyparam}[1]{\mathit{#1}}
\newcommand{\lockinterp}{\thyparam{lock{\text -}interp}}
\newcommand{\abssteps}{\thyparam{abs{\text -}steps}}
\newcommand{\thyvar}[1]{\mathit{#1}}
\newcommand{\lc}{\thyvar{lc}}
\newcommand{\pc}{\thyvar{pc}}
\newcommand{\tps}{\thyvar{tps}}
\newcommand{\regs}{\thyvar{regs}}
\newcommand{\mds}{\thyvar{mds}}
\newcommand{\mem}{\thyvar{mem}}
\newcommand{\failed}{\thyvar{failed}}
\newcommand{\PCs}{\thyvar{PCs}}
\newcommand{\Cs}{\thyvar{Cs}}
\newcommand{\nl}{\thyvar{nl}}

\newcommand{\unused}[1]{{\color{gray!80}#1}}
\newcommand{\defined}[1]{\mathsf{#1}}
\newcommand{\CouplingInvPres}{\defined{coupling{\text -}inv{\text -}pres}}
\newcommand{\ComSecure}{\defined{com{\text -}secure}}
\newcommand{\SysSecure}{\defined{sys{\text -}secure}}

\newcommand{\StrongLowBisimMM}{\defined{strong{\text -}low{\text -}bisim{\text -}mm}}
\newcommand{\NoHighBranching}{\defined{no{\text -}high{\text -}branching}}
\newcommand{\SecureRefinement}{\defined{secure{\text -}refinement}}
\newcommand{\SimplerRefinementSafe}{\defined{decomp{\text -}refinement{\text -}safe}}
\newcommand{\CgConsistent}{\defined{cg{\text -}consistent}}
\newcommand{\SecureRefineSimpler}{\defined{secure{\text -}refinement{\text -}decomp}}
\newcommand{\ClosedOthers}{\defined{closed{\text -}others}}
\newcommand{\PreservesMM}{\defined{preserves{\text -}modes{\text -}mem}}
\newcommand{\High}{\defined{High}}
\newcommand{\Low}{\defined{Low}}
\newcommand{\Writable}{\defined{writable}}
\newcommand{\Readable}{\defined{readable}}
\newcommand{\Sym}{\defined{sym}}
\newcommand{\CompileExpr}{\defined{compile{\text -}expr}}
\newcommand{\CompileCmd}{\defined{compile{\text -}cmd}}
\newcommand{\CompilerInputReqs}{\defined{compile{\text -}cmd{\text -}input{\text -}reqs}}
\newcommand{\CompiledCmdConfigConsistent}{\defined{compiled{\text -}cmd{\text -}config{\text -}consistent}}
\newcommand{\stops}{\defined{stops}}
\newcommand{\AsmrecMdsConsistent}{\defined{asmrec{\text -}mds{\text -}consistent}}
\newcommand{\RegrecMemConsistent}{\defined{regrec{\text -}mem{\text -}consistent}}
\newcommand{\VarStable}{\defined{var{\text -}stable}}
\newcommand{\RegrecStable}{\defined{regrec{\text -}stable}}
\newcommand{\NoUnstableExprs}{\defined{no{\text -}unstable{\text -}exprs}}
\newcommand{\AbsStepsWR}{\defined{abs{\text -}steps}_\mathsf{wr}}
\newcommand{\RefRelWR}{\mathcal{R}_\mathsf{wr}}
\newcommand{\CouplInvWR}{\mathcal{I}_\mathsf{wr}}
\newcommand{\BCof}{\mathcal{B}\defined{_Cof}}
\newcommand{\BCofApplied}{\BCof~\mathcal{B}~\mathcal{R}~\mathcal{I}}
\newcommand{\Regrec}{\defined{regrec}}
\newcommand{\Asmrec}{\defined{asmrec}}
\newcommand{\fst}{\defined{fst}}
\newcommand{\snd}{\defined{snd}}
\newcommand{\map}{\defined{map}}
\newcommand{\mapfst}{\map\ \fst}
\newcommand{\mapsnd}{\map\ \snd}
\newcommand{\False}{\defined{False}}
\newcommand{\True}{\defined{True}}
\newcommand{\Suc}{\defined{Suc}}
\newcommand{\Joinable}{\defined{joinable}}
\newcommand{\JoinableFwd}{\defined{joinable{\text -}forward}}
\newcommand{\JoinableBwd}{\defined{joinable{\text -}backward}}
\newcommand{\length}{\defined{length}}
\newcommand{\None}{\defined{None}}
\newcommand{\Some}{\defined{Some}}
\newcommand{\Cset}{\mathcal{C}}
\newcommand{\Cvars}{\mathcal{C}\defined{vars}}
\newcommand{\dmaFunc}{\mathcal{L}}
\newcommand{\dmaApp}[2]{\dmaFunc\ {#1}\ #2}
\newcommand{\LowMdsEqOp}[1]{=_{#1}^{\Low}}
\newcommand{\LowMdsEq}[3]{#2 \LowMdsEqOp{#1} #3}
\newcommand{\LCLowMdsEqOp}{=_{\defined{mds}}^{\Low}}
\newcommand{\LCLowMdsEq}[2]{#1 \LCLowMdsEqOp #2}
\newcommand{\LCSameMdsOp}{=_{\defined{mds}}}
\newcommand{\LCSameMdsMemOp}{\LCSameMdsOp^{\defined{mem}}}
\newcommand{\LCSameMds}[2]{#1 \LCSameMdsOp #2}
\newcommand{\LCSameMdsMem}[2]{#1 \LCSameMdsMemOp #2}
\usepackage{xspace} % to fix spacing in macros producing words.
\newcommand{\Covern}{\textsc{Covern}\xspace}
\newcommand{\WRCompiler}{\textsf{wr-compiler}\xspace}
\usepackage{mathpartir}
\usepackage{tikz}
\ifPreprint
\usepackage{fancyhdr}
\fi

\begin{document}

%% Capitalisation of cross references
\renewcommand{\sectionautorefname}{Section}
\renewcommand{\subsectionautorefname}{Section}
\renewcommand{\subsubsectionautorefname}{Section}

\maketitle

\ifPreprint
\hideLIPIcs
\pagestyle{fancyplain}
\thispagestyle{fancyplain}
\fancyhead[C]{Extended version: 1st July 2019. To appear in ITP 2019}
\fancyhead[L]{}
\fancyhead[R]{}
\fancyfoot[C]{\vspace*{2.5\baselineskip}\thepage}
\fi

\begin{abstract}
It is common to prove by reasoning over source code that programs do not leak sensitive data.
But doing so leaves a gap between
  reasoning and reality that can only be filled by accounting for the
  behaviour of the compiler.
This task is complicated when programs enforce \emph{value-dependent} information-flow security properties---in which classification of locations can vary depending on values in other locations---and complicated further when programs exploit shared-variable concurrency.

  Prior work has formally defined a notion of concurrency-aware refinement
  for preserving value-dependent security properties.
However, that notion is considerably more complex than standard refinement definitions typically applied in the verification of semantics preservation by compilers.
To date it remains unclear whether it can be applied to a realistic compiler,
because there exist no general decomposition principles for separating it into
smaller, more familiar, proof obligations.

In this work, we provide such a decomposition principle, which we show
  can almost halve the complexity of proving secure refinement. Further,
  we demonstrate its applicability to secure compilation, by proving in Isabelle/HOL the preservation of value-dependent
  security by a proof-of-concept compiler from an imperative While language to
  a generic RISC-style assembly language, for programs with shared-memory
  concurrency mediated by locking primitives.
  Finally, we execute our compiler in Isabelle on a While language model of the Cross Domain Desktop Compositor, demonstrating to our knowledge the first use of a compiler verification result to carry an information-flow security property down to the assembly-level model of a non-trivial concurrent program.

\end{abstract}

\section{Introduction}

It is well known that program translations of the kind carried out by
compilers can in principle break security properties like
confidentiality~\cite{Kaufmann16,Barthe18}.
Yet source level reasoning about confidentiality remains
common~\cite{Murray_SE_18,Mantel_SS_11,Lourenco_Caires_15}. Existing verified
compilers like CompCert \cite{Leroy09} and CakeML \cite{Kumar_MNO_14}
preserve semantics, but semantics preservation alone may be insufficient to
preserve confidentiality, especially for shared memory concurrent programs
whose threads must guard against timing leaks in order to prevent them manifesting as storage leaks~\cite{Murray_SPR_16}.

Supporting secure compilation of
programs that must enforce \emph{value-dependent} security policies
poses an additional challenge, because in such policies the sensitivity
of a memory location can depend on the values held in other memory locations.
Thus, unlike prior work on secure compilation~\cite{Barthe10},
preserving security under refinement requires a refinement relation that is strong enough to preserve those memory contents on which the policy depends.

In prior work~\cite{Murray_SPR_16}, we presented a definition for a notion
of value-dependent security-preserving refinement that is compositional
for concurrent
programs: by applying it to each thread individually,
one can derive a secure refinement of the concurrent composition.

\begin{figure}
    \begin{subfigure}[l]{0.56\textwidth}
        {\setlength{\mathindent}{0cm}\footnotesize
        \begin{align*}
            & \CouplingInvPres\ \mathcal{B}\ \mathcal{R}\ \mathcal{I}\ \equiv \\
            & \forall \lc_{1A}\ \lc_{1C}. \ 
              (\lc_{1A}, \lc_{1C}) \in \mathcal{R} \longrightarrow \\
            & \quad (\forall \lc_{1C}'. \ 
                     \lc_{1C} \EvalStepConc \lc_{1C}' \longrightarrow \\
            & \qquad (\exists n\ \lc_{1A}'.\ 
                     \lc_{1A} \NEvalStepAbs{n} \lc_{1A}'\ \land\
                     (\lc_{1A}', \lc_{1C}') \in \mathcal{R}\ \land \\
            & \qquad \quad (\forall \lc_{2A}\:\lc_{2C}\:\lc_{2A}'.\,
                             (\lc_{1A}, \lc_{2A}) \in \mathcal{B}\land
                             \LCSameMds{\lc_{1A}}{\lc_{2A}}\,\land \\
            & \qquad \qquad (\lc_{2A}, \lc_{2C}) \in \mathcal{R}\,\land \ 
                            (\lc_{1C}, \lc_{2C}) \in \mathcal{I}\ \land \\
            & \qquad \qquad \LCSameMds{\lc_{1C}}{\lc_{2C}}\,\land\,
                            \lc_{2A} \NEvalStepAbs{n} \lc_{2A}'\,\land\,
                            \LCSameMds{\lc_{1A}'}{\lc_{2A}'} \\
            & \qquad \qquad \longrightarrow \,
                            (\exists \lc_{2C}'.\ 
                                \lc_{2C} \EvalStepConc \lc_{2C}'\ \land\ 
                                \LCSameMds{\lc_{1C}'}{\lc_{2C}'}\ \land\\
            & \qquad \qquad \qquad \quad \ (\lc_{2A}', \lc_{2C}') \in \mathcal{R}\ \land \ 
                                           (\lc_{1C}', \lc_{2C}') \in \mathcal{I}))))
        \end{align*}}
    \end{subfigure}
    \begin{subfigure}[l]{0.43\textwidth}
        % Pull into here because the publisher mandates
        % only a single .tex file per article. -robs.
        %\input{imgs/coupling_invariant.tex}
        % Reproduced from our CSF'16 paper repo. Originally by Toby Murray
        \begin{tikzpicture}[scale=0.7]
            \node at (0,0) {\footnotesize{$1A$}};
            \draw [->, dashed] (0.4,0) -- (3.6,0);;
            \node at (2,-0.3) {\footnotesize$n$};
            \node at (4,0) {\footnotesize{$1A'$}};

            \node at (0.5,1) {\footnotesize{$2A$}};
            \draw [->] (0.9,1) -- (4.1,1);;
            \node at (2.5,0.7) {\footnotesize$n$};
            \node at (4.5,1) {\footnotesize{$2A'$}};

            \node at (0.25,0.5) {$\mathcal{B}$};
            \node at (4.25,0.5) {$\mathcal{B}$};

            \node at (0.25,-3.5) {$\mathcal{I}$};
            \node at (4.25,-3.5) {$\mathcal{I}$};

            \node at (-1.85,0.5) {\begin{minipage}{2cm}\centering\ \ \ \ \  {\em abstract execution}\end{minipage}};

            \node at (-1.85,-3.5) {\begin{minipage}{2cm}\centering\ \ \ \ \ {\em concrete execution}\end{minipage}};

            \draw (0,-0.5) -- (0,-1.5);
            \node at (0,-2) {$\mathcal{R}$};
            \draw (0,-2.5) -- (0,-3.5);

            \draw (0.5,0.5) -- (0.5,-0.5);
            \node at (0.5,-1) {$\mathcal{R}$};
            \draw (0.5,-1.5) -- (0.5,-2.5);

            \draw [dashed] (4,-0.5) -- (4,-1.5);
            \node at (4,-2) {$\mathcal{R}$};
            \draw [dashed] (4,-2.5) -- (4,-3.5);

            \draw [dashed] (4.5,0.5) -- (4.5,-0.5);
            \node at (4.5,-1) {$\mathcal{R}$};
            \draw [dashed] (4.5,-1.5) -- (4.5,-2.5);

            \node at (0,-4) {\footnotesize{$1C$}};
            \draw [->] (0.4,-4) -- (3.6,-4);;
            \node at (2,-4.3) {\footnotesize$1$};
            \node at (4,-4) {\footnotesize{$1C'$}};

            \node at (0.5,-3) {\footnotesize{$2C$}};
            \draw [->, dashed] (0.9,-3) -- (4.1,-3);;
            \node at (2.5,-3.3) {\footnotesize$1$};
            \node at (4.5,-3) {\footnotesize{$2C'$}};
        \end{tikzpicture}
    \end{subfigure}
    \vspace{-1\baselineskip}
    \caption{Definition, graphical depiction of refinement preservation for $\SecureRefinement$ (Def. \ref{def:secure-refinement})}
    \label{fig:coupling-inv-pres}
\end{figure}
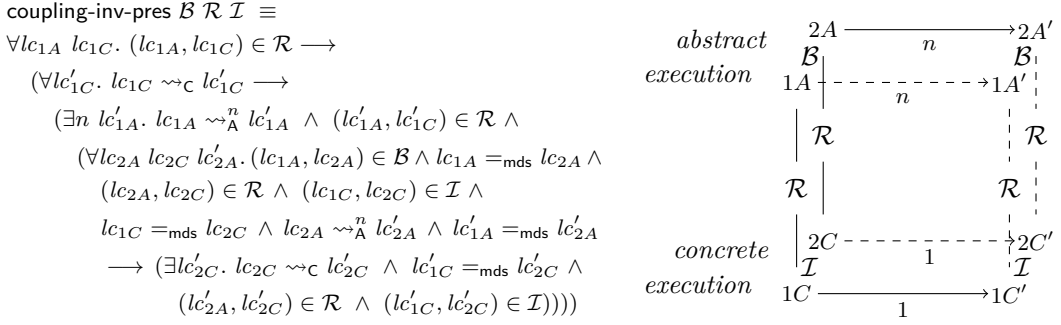

The essence of this notion of security-preserving refinement (presented fully in \autoref{sec:cvdni-refinement}) is in its refinement preservation obligation ($\CouplingInvPres$ in Figure~\ref{fig:coupling-inv-pres}).
Here, the usual square-shaped commuting diagram that is commonly used to depict (semantics-preserving) refinement (Figure~\ref{fig:decomp-R}) has been replaced by a \emph{cube} (Figure~\ref{fig:coupling-inv-pres}).
The additional dimension of this cube reflects that it preserves a 2-safety hyperproperty~\cite{Clarkson10} that compares two executions rather than examining a single one.
As such, it is significantly more complicated to prove than standard notions of
semantics-preserving refinement typical in verified
compilation~\cite{Leroy09,Kumar_MNO_14}.

To date there exist no verified compilers for shared-variable concurrent
programs proved to preserve value-dependent information-flow security.
We argue that without a \emph{decomposition principle} the cube-shaped refinement notion is too cumbersome to prove for realistic compilers.

In this paper, we tackle the central problem of making our notion of secure refinement applicable to verified secure compilation.
Firstly, we present a decomposition principle that makes the cube-shaped notion more tractable.
Secondly, we demonstrate its tractability with our major contribution: a machine-checked formal proof of concurrent value-dependent security preservation, for a proof-of-concept compiler.

In \autoref{sec-simpler} we present our decomposition principle, which decomposes the cube (\autoref{fig:coupling-inv-pres})
into three separate obligations (\autoref{fig:decomp}). The first
of these is akin to semantics-preserving refinement, while the second
and third essentially ensure together that the refinement has not
introduced any termination- and timing-leaks.

In \autoref{sec-h-branch-example} we show how the decomposition principle can almost halve the effort to prove secure refinement -- in this case, of a program that is especially prone to introduced timing leaks because it branches on secrets (a feature not yet allowed by our compiler).
There, we present a side-by-side comparison of the proof effort, both with and without the decomposition principle.
We find that using it reduces the proof's complexity by 44\%.

In \autoref{sec-compiler}, we present our compiler and its formal verification, as an application of the decomposition principle.
This compiler translates concurrent programs written in an imperative While
language, with locking primitives for mediating access to shared memory,
into a RISC-style assembly language. It does so by compiling each thread
individually, and in doing so preserves a formal security property that remains compositional between threads.
Furthermore, our compiler demonstrates a way of formalising and proving when it is safe for a compiler to perform optimisations in the presence of concurrency.
To ensure that the contents of shared memory locations are preserved under compilation despite potential interference from other threads, our compiler tracks which shared memory locations are \emph{stable} (free from any such interference).
It then makes use of this tracking to avoid redundant loads from stable shared variables safely, that would otherwise be considered unsafe to omit.

All results are mechanised in Isabelle/HOL,\footnote{The $\WRCompiler$ totals $\sim$7k lines, and verification + compilation of the 2-thread CDDC model totals $\sim$1.6k lines of Isabelle proof script, excluding whitespace and comments. See ``Supplement Material''.}
and in \autoref{sec:exec-instantiation} we explain how, in order to validate our theory, we instantiated it so that we could execute our compiler in Isabelle.
This enabled us to execute it over a While language model of the Cross Domain Desktop Compositor \cite{Beaumont_MM_16} (CDDC), a concurrent program that enforces information flow control over value-dependently classified input.
To our knowledge this is the first proof of information flow security for an assembly-level model of a non-trivial concurrent program, demonstrating the power of verified secure compilation for deriving security properties of compiled code.

\section{Background and example}\label{sec-background}

We begin by introducing with an illustrative example (\autoref{fig:example-worker}) the challenges of verifying \emph{value-dependent information-flow security} in the presence of \emph{shared-variable concurrency}.

% specific
Consider the task of verifying a multithreaded system that manages the user interface (UI) for a \emph{dual-personality smartphone}, a phone that provides clearly distinguished user contexts (\emph{personalities}), typically for work versus leisure.
Specifically, our task is to verify that it does not leak \emph{sensitive} information intended only for one of those personalities, which we classify $\High$ (\autoref{fig:example-worker-dma-high}), to locations belonging to the other, which we classify $\Low$ (\autoref{fig:example-worker-dma-low}).

% general
Here and generally, our \emph{attacker model} is an entity that can read from the system's \emph{untrusted sinks}: some subset of permanently $\Low$-classified locations not subject to synchronisation.
% specific
In our example, this may include
WLAN device registers in a hostile environment.

% specific
The smartphone's UI system consists of a number of threads running concurrently with a shared address space, and we aim to verify that as a whole it satisfies the security requirement.
% general
But to avoid a state space explosion that is exponential in the number of threads, we must do this \emph{compositionally}: one thread at a time, then combining the results of these analyses.

% specific
We focus on a particular worker thread (\autoref{fig:example-worker-code}), the one responsible for sending touchscreen input from the \var{source} variable to its intended destination.

The first challenge is that the destination depends on which personality the phone is currently providing, which is indicated by the value of \var{domain}.
This is reflected by the classification of \var{source} being dependent on the value of \var{domain}: \var{source} is classified $\Low$ exactly when $\var{domain} = \const{LOW}$ (where $\const{LOW}$ is a designated constant), and is classified $\High$ otherwise.
Due to this dependency, \var{domain} is known as a \emph{control variable} of \var{source}.

The second challenge is the worker thread runs in a shared address space that might be accessed or modified by other threads, for various purposes.
One of these threads may be responsible for maintaining that $\var{domain} = \const{LOW}$ exactly when the phone indicates it is providing the $\Low$ personality (\autoref{fig:example-worker-dma-low}), so the user knows not to type in anything sensitive.
Another thread may be responsible for assigning $\Assign{\var{suspended}}{\const{TRUE}}$ when the user turns the phone's screen off, to make the worker stop processing touchscreen input.
We may then wish for \var{workspace} to be usable by some other thread---e.g.~processing input from a fingerprint scanner---in such a way that it can assume \var{workspace} no longer contains any sensitive values.

% general
When we analyse one thread like this worker in terms of our compositional security property (\autoref{sec:cvdni-security}), all of the other threads in the system are trusted to do two things:
\begin{enumerate}
    \item They follow a \emph{synchronisation scheme}:
        % specific
        here, if read- or write-access to a certain variable is governed by a lock, they must hold it in order to access the variable in that manner.
    % general
    \item They themselves do not leak values from $\High$-classified locations
        (we refer to such values themselves as $\High$)
        to $\Low$-classified locations that are read-accessible to other threads.
        Note we are proving that the thread we are analysing can be trusted in the same way.
\end{enumerate}
Even under these assumptions, the concurrency gives rise to some tricky considerations.

% general
Firstly, it is important that no thread in the system (including the thread under analysis) modifies any control variables carelessly.
% specific
For example, writing $\Assign{\var{domain}}{\const{LOW}}$ immediately after the worker reads a $\High$ value from \var{source}, will cause it to leak to \var{low\_sink}.
To prevent this, the worker uses \var{source\_lock}, granting it \emph{exclusive write-access} to \var{source} and \var{domain}.

% general
Furthermore as noted above, we may want to ensure that a \emph{non-attacker-observable} location is nevertheless cleared of any sensitive values before being used by another thread.
% specific
In our example, we classify \var{workspace} $\Low$ for the analysis to enforce this when the worker is suspended, but as the worker sometimes uses it to process $\High$ values, it is important to know \var{workspace} is accessible only to the worker during that time.
To ensure this, the worker uses \var{workspace\_lock}, granting it \emph{exclusive read- and write-access} to \var{workspace}.
It is then responsible for clearing it of any $\High$ values by the time it releases exclusive read-access.

\begin{figure}
    \begin{subfigure}{0.5\textwidth}
{\footnotesize$
\While{\const{TRUE}}{
    \Seq{\myindent{1}\LockAcq{\var{workspace\_lock}}}{
    \Seq{\Whilei{1}{!\var{suspended}}{
    \Seq{\myindent{2}\LockAcq{\var{source\_lock}}}{
    \Seq{\myindent{2}\Assign{\var{workspace}}{\var{source}}}{
        \myindent{2}\text{/* \ldots\ operations on \var{workspace} \ldots */}\\
    \Seq{\ITEi{2}{\var{domain} = \const{LOW}}{
            \myindent{3}\Assign{\var{low\_sink}}{\var{workspace}}
        }{
            \Seq{\myindent{3}\Assign{\var{high\_sink}}{\var{workspace}}}{
            \myindent{3}\Assign{\var{workspace}}{\const{0}}}
        }}{
        \myindent{2}\LockRel{\var{source\_lock}}}}}
    }}{
    \myindent{1}\Seq{\LockRel{\var{workspace\_lock}}}{
    \myindent{1}\Whileg{\var{suspended}}{\Skip}}}}
}
$}
        \caption{Input processing worker thread program}
        \label{fig:example-worker-code}
    \end{subfigure}
    \begin{subfigure}{0.46\textwidth}
        \centering
        \includegraphics[height=0.1\textheight]{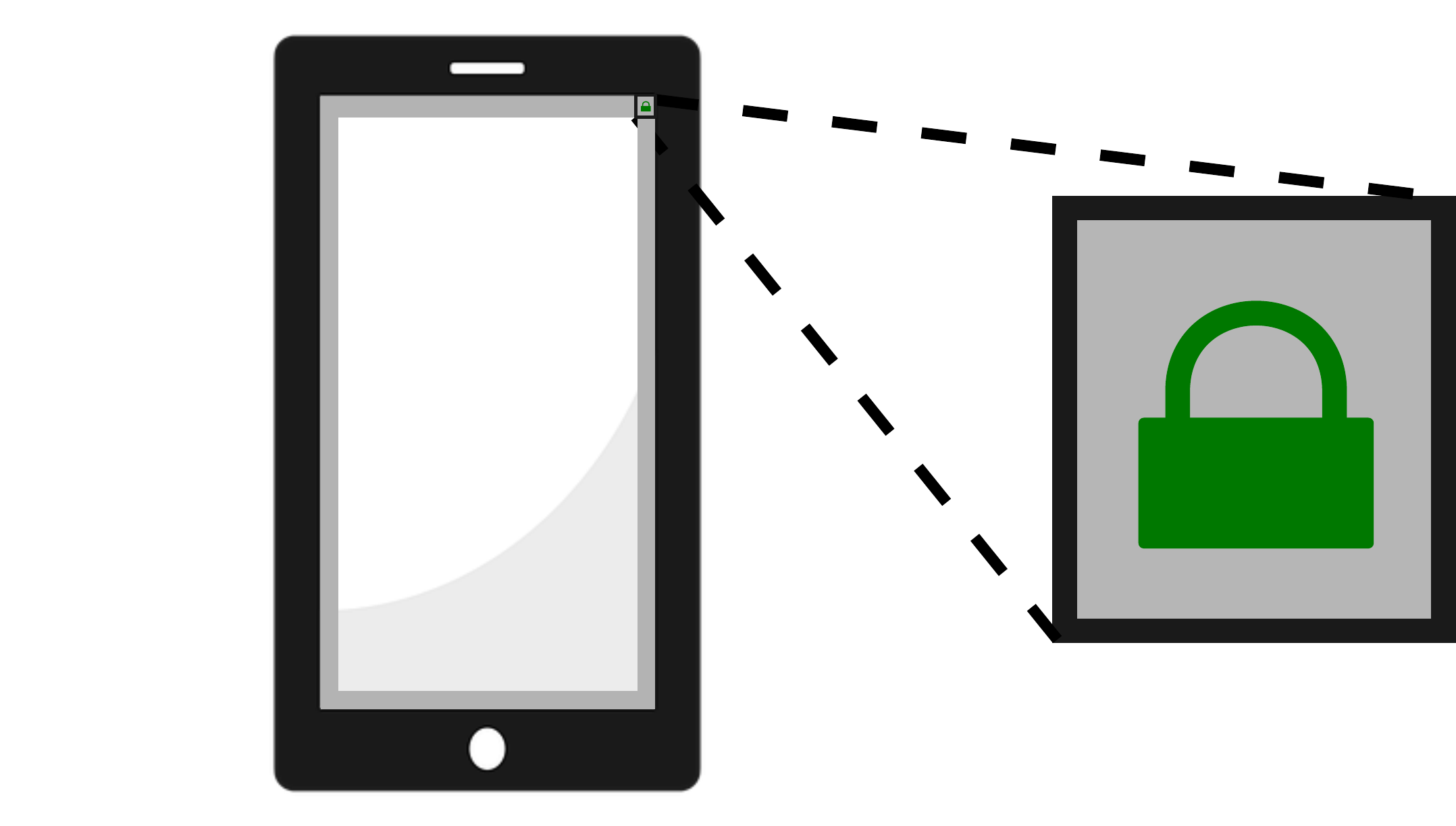}
        \caption{The phone providing the $\High$ personality:
                 $\var{domain} \neq \const{LOW}$, and
                 $\var{source}$ is classified $\High$
                 to reflect that the user might type in secrets.}
        \label{fig:example-worker-dma-high}

        \vspace{0.5\baselineskip}

        \includegraphics[height=0.1\textheight]{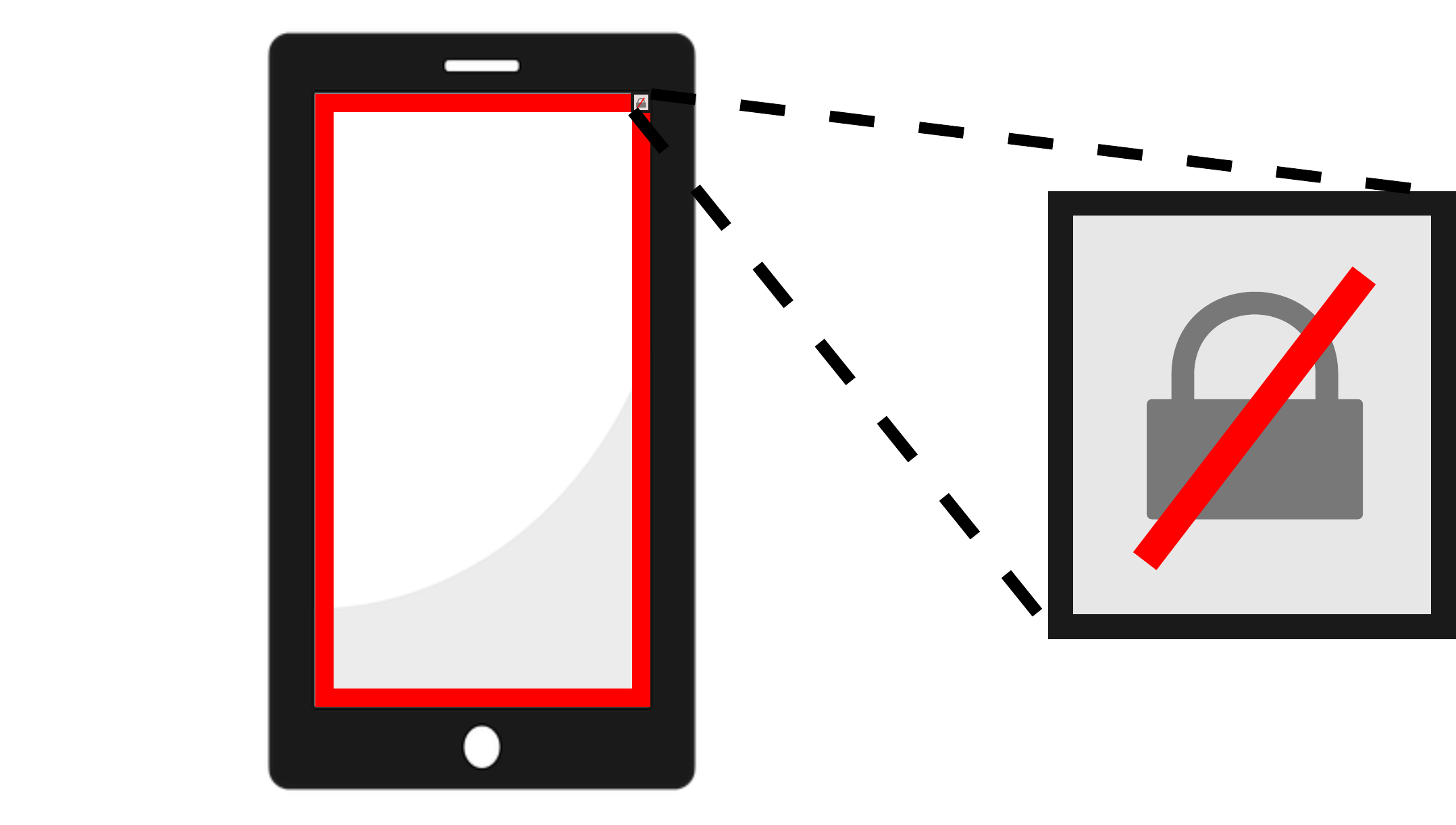}
        \caption{The phone displaying visual indicators that it is providing the $\Low$ personality:
                 $\var{domain} = \const{LOW}$, and
                 $\var{source}$ is classified $\Low$
                 to reflect that we trust the user not to type in secrets.}
        \label{fig:example-worker-dma-low}
    \end{subfigure}
    \caption{Example: Touchscreen input processing for a dual-personality smartphone}
    \label{fig:example-worker}
\end{figure}

\subsection{Concurrent value-dependent noninterference (CVDNI)}\label{sec:cvdni-security}

Having illustrated the challenges with an example, we now focus on the formalisation of our information-flow security property CVDNI, which we target with our per-thread analysis, and which our compiler preserves.
It is defined in terms of two main elements:
\begin{enumerate}
    \item a binary \emph{strong low-bisimulation (modulo modes)} relation $\mathcal{B}$ between program configurations, that establishes the required information-flow security property.
        Like Goguen \& Meseguer-style noninterference \cite{Goguen_Meseguer_82}, any states it relates must agree on their ``low'' portions, and it demands that lock-step execution preserve that correspondence.
        This section will explain how it is specialised further for shared-variable concurrency.
    \item a \emph{classification} function $\dmaFunc$ that determines the ``low'' portion of a program configuration, thus affecting $\mathcal{B}$'s requirements.
        Unlike \cite{Goguen_Meseguer_82} however, $\dmaFunc$ here can depend on values in the program configuration itself, thus expressing dynamic and not just static classifications.
\end{enumerate}

We now present definitions from Section III-2b of our previous work \cite{Murray_SPR_16} simplified as noted.
The theory is parameterised over the type of values~$\Val$, a finite set of shared variables~$\Var$, and a \emph{deterministic evaluation step semantics} $\EvalStep$ between \emph{local configurations} (of a thread in a concurrent program) each denoted by a triple $\LocalConf{\tps}{\mds}{\mem}$:
\begin{itemize}
    \item $\tps$ is the \emph{thread-private state}, which is permanently inaccessible to the attacker and the other threads.
        Note that due to this inaccessibility, we allow the user of the theory to parameterise the type of $\tps$, and do not impose any particular structure.
    \item $\mds \oftype \Mode \Rightarrow \Var\ \set$
      is the \emph{(access) mode state}, which is ghost state associating each
      $\Mode = \{\mathbf{AsmNoW},\mathbf{AsmNoRW},\mathbf{GuarNoW},\mathbf{GuarNoRW}\}$
      with a set of shared variables.
      Intuitively, it identifies the set of variables for which the thread currently possesses (or respects) a kind of exclusivity of access granted (or obligated) by a synchronisation scheme.
      This facilitates compositional, assume-guarantee~\cite{Jones:phd} style reasoning.
      For example, when our worker thread holds \var{source\_lock}, it \emph{assumes no other threads write} to \var{source} or its control variable ($\{\var{source}, \var{domain}\} \subseteq \mds\ \mathbf{AsmNoW}$), otherwise it \emph{guarantees it does not write} to them ($\mathbf{GuarNoW}$).
      Similarly, holding \var{workspace\_lock} it assumes no other threads \emph{read or write} to \var{workspace} ($\var{workspace} \in \mds\ \mathbf{AsmNoRW}$), and at all other times it makes the corresponding guarantee ($\mathbf{GuarNoRW}$).
    \item $\mem \oftype \Mem$ is \emph{shared memory} considered potentially accessible to the attacker and other threads.
        In order to make what is accessible amenable to analysis, we impose the structure $\Mem = \Var \Rightarrow \Val$, a total map from shared variable names to their values.
\end{itemize}
The theory is then further parameterised by the value-dependent classification function $\dmaFunc \oftype \Mem \Rightarrow \Var \Rightarrow \{\High, \Low\}$, and a function $\Cvars \oftype \Var \Rightarrow \Var\ \set$ that returns all the control variables of a given variable.
In our worker thread example, $\dmaApp{\mem}{x}$ gives:
\begin{itemize}
    \item $\High$ when $x$ is \var{high\_sink}, meaning \var{high\_sink} is classified $\High$ at all times.
    \item when $x$ is \var{source}: $\Low$ if $\mem\ \var{domain}=\const{LOW}$, and $\High$ otherwise.
    \item $\Low$ for all other variables $x$, meaning they are classified $\Low$ at all times.
\end{itemize}
The set $\Cset = \{y\ |\ \exists x.\ y \in \Cvars\ x\}$ is then defined to contain all control variables in the system.
Thus in our worker thread example, $\Cvars\ \var{source} = \{\var{domain}\}$ and $\Cset = \{\var{domain}\}$.

To support compositionality for concurrent programs, the ``low'' portion demanded to be equal by the analysis is tightened up to be \emph{modulo modes} -- it includes non-control variables only if they are assumed to be \emph{readable} by other threads according to the mode state:
$\Readable\ \mds\ x \equiv x \notin \mds\ \mathbf{AsmNoRW}$.
Thus intuitively, the user of the theory should model permanent untrusted output sinks of the whole concurrent program, as variables for which $\dmaFunc$ \emph{always returns $\Low$}, ungoverned by any synchronisation scheme that the attacker cannot be trusted to follow.
(In our example, \var{low\_sink} is untrusted permanently in this way, but \var{workspace} is untrusted only when unlocked.)
The notion of observational indistinguishability used for the noninterference property is then defined over memories as follows.

\begin{definition} [Low-equivalent memories modulo modes] \label{def:low_mds_eq-mem}
\begin{align*}
    & \LowMdsEq{\mds}{\mem_1}{\mem_2}\ \equiv \\
    & \forall x.\ x \in \Cset\ \lor\ \dmaApp{\mem_{1}}{x} = \Low\ \land\ \Readable\ \mds\ x\ \longrightarrow \
      mem_1\ x = \mem_2\ x
\end{align*}
\end{definition}
For this paper, we will use notation $\lc_1 \LCLowMdsEqOp \lc_2$ to lift $\LowMdsEqOp{\mds}$ to local program configurations, asserting also that $\lc_1$ and $\lc_2$ are \emph{modes-equal} (have the same mode state).
Additionally, we will use notation $\LCSameMds{\lc_1}{\lc_2}$ to denote (alone) that $\lc_1$ and $\lc_2$ are modes-equal.

The per-thread compositional security property $\ComSecure$ asserts the existence of a witness relation $\mathcal{B}$ for every possible observationally equivalent pair of starting configurations:

\begin{definition} [Per-thread compositional CVDNI property] \label{def:com-secure}
\begin{align*}
    & \ComSecure\ (\tps, \mds)\ \equiv \ 
      \forall mem_1\ \mem_2.\ \LowMdsEq{\mds}{\mem_1}{\mem_2} \longrightarrow \\
    & \qquad (\exists \mathcal{B}.\ 
              \StrongLowBisimMM\ \mathcal{B}\ \land \
              (\LocalConf{\tps}{\mds}{\mem_1},
               \LocalConf{\tps}{\mds}{\mem_2}) \in \mathcal{B})
\end{align*}
\end{definition}
where all such witness relations $\mathcal{B}$ must be a \emph{strong low-bisimulation (modulo modes)}:
\begin{align*}
    & \StrongLowBisimMM\ \mathcal{B}\ \equiv \ 
      \CgConsistent\ \mathcal{B}\ \land \ 
      \Sym\ \mathcal{B}\ \land \\
    & (\forall \lc_1\ \lc_2.\ (\lc_1, \lc_2) \in \mathcal{B}\ \land \ 
      \LCSameMds{\lc_1}{\lc_2} \longrightarrow \ 
      \LCLowMdsEq{\lc_1}{\lc_2}\ \land \\
    & \qquad (\forall \lc_1'.\ \lc_1 \EvalStep \lc_1' \longrightarrow \ 
             (\exists \lc_2'.\ \lc_2 \EvalStep \lc_2'\ \land \ 
              \LCSameMds{\lc_1'}{\lc_2'}\ \land \ 
              (\lc_1', \lc_2') \in \mathcal{B})))
\end{align*}

That is, $\mathcal{B}$ must maintain observational indistinguishability by requiring that all configuration pairs it relates that have the same mode state, are low-equivalent modulo modes.

Furthermore, it must be a \emph{bisimulation} by being symmetric and \emph{progressing to itself}: any step taken by one of the configurations must be able to be matched by a step taken by the configuration related to it, such that the destinations remain related by $\mathcal{B}$ (and modes-equal).

Finally---and the most crucial element ensuring the property's compositionality for concurrent programs---is the condition that $\mathcal{B}$ must be $\CgConsistent$: \emph{closed under globally consistent changes} made to memory by other threads,
which is to say, changes that preserve low-equivalence and are permitted by the current mode state $\mds$.
Specifically, the environment (of other threads) is permitted to change either of variable~$x$'s value or its classification only when~$x$ is \emph{writable}:
$\Writable\ \mds\ x \equiv x \notin \mds\ \mathbf{AsmNoW}\ \land\ x \notin \mds\ \mathbf{AsmNoRW}$.
\begin{definition} [Closedness under globally consistent changes] \label{def:cg-consistent}
\begin{align*}
    & \CgConsistent\ \mathcal{B}\ \equiv \ 
      \forall \tps_1\ \mem_1\ \tps_2\ \mem_2\ \mds. \\
    & \qquad (\LocalConf{\tps_1}{\mds}{\mem_1},
              \LocalConf{\tps_2}{\mds}{\mem_2}) \in \mathcal{B}\ \longrightarrow \\
    & \qquad (\forall \mem_1'\ \mem_2'.\ 
              (\forall x.\ (\mem_1\ x \ne \mem_1'\ x\ \lor \ 
                            \mem_2\ x \ne \mem_2'\ x\ \lor \\
    & \qquad \ \ \dmaApp{\mem_1}{x} \ne \dmaApp{\mem_1'}{x}) \ 
                 \longrightarrow\ \Writable\ \mds\ x)\ \land \ 
                 \LowMdsEq{\mds}{\mem_1'}{\mem_2'}\ \longrightarrow \\
    & \qquad (\LocalConf{\tps_1}{\mds}{\mem_1'},
              \LocalConf{\tps_2}{\mds}{\mem_2'}) \in \mathcal{B})
\end{align*}
\end{definition}

Theorem 3.1 of our prior work \cite{Murray_SPR_16} then gives us that the parallel composition
of $\ComSecure$ programs is itself a program that enforces a system-wide value-dependent noninterference property ($\SysSecure$, for whose details we refer the reader to Section III-2(a) of~\cite{Murray_SPR_16}).

\begin{figure}
    \begin{subfigure}[t]{0.33\textwidth}
        \begin{tikzpicture}[overlay]
            \node[rounded corners,anchor=west,fill=gray!55,minimum width=1.2cm,minimum height=\baselineskip-0.05cm] (RrelAif) at (-0.25em,0.1cm) {};
            \node[rounded corners,anchor=west,fill=gray!20,minimum width=1.1cm,minimum height=\baselineskip-0.05cm] (RrelAt) at (1.25em,0.1cm-\baselineskip) {};
            \node[rounded corners,anchor=west,fill=gray!90,minimum width=1.6cm,minimum height=\baselineskip-0.05cm] (RrelAe) at (1.25em,0.1cm-3\baselineskip) {};
        \end{tikzpicture}%
        {\footnotesize$
        \ITE{\var{h} \neq \const{0}}{
            \myindent{1}\Assign{\var{x}}{\var{y}}
        }{
            \myindent{1}\Assign{\var{x}}{\var{y} + \var{z}}
        }$}
        \caption{Abstract if-conditional. \\
        Relation $\mathcal{R}$ pairs configurations of this program
        with configurations of the program
        in \autoref{fig:refinement-example-conc}
        that are of the same-shaded region.}
        \label{fig:refinement-example-abs}
    \end{subfigure}
    \hfill
    \newcommand{\ITEthen}[2]{\kw{if}\ #1\ \kw{then}\\#2}
    \newcommand{\ITEelse}[1]{\kw{else}\\#1\\\kw{fi}}
    \begin{subfigure}[t]{0.61\textwidth}
    \begin{subfigure}[t]{0.45\textwidth}
        \begin{tikzpicture}[overlay]
            \node[rounded corners,anchor=west,fill=gray!55,minimum width=1.6cm,minimum height=2\baselineskip-0.05cm] (RrelCif) at (-0.25em,0.1cm-0.5\baselineskip) {};
            \node[rounded corners,anchor=west,fill=gray!20,minimum width=1.6cm,minimum height=4\baselineskip-0.05cm] (RrelCt) at (1.5em,0.1cm-3.5\baselineskip) {};
            \node[rounded corners,anchor=west,fill=gray!90,minimum width=3cm,minimum height=4\baselineskip-0.05cm] (RrelCe) at (14.9em,0.1cm-3.5\baselineskip) {};
        \end{tikzpicture}%
        {\footnotesize$
        \Seq{\Assign{\var{reg3}}{h}}{
        % Compensate for weird offset created by tikz overlay -robs.
        \vspace{-0.33\baselineskip}
        \ITEthen{\var{reg3} \neq \const{0}}{
            \myindent{1}\tikz[overlay, remember picture]\node[anchor=west](Irel1t){\Skip;} ; \\
            \myindent{1}\tikz[overlay, remember picture]\node[anchor=west](Irel2t){\Skip;} ; \\
            \myindent{1}\tikz[overlay, remember picture]\node[anchor=west](Irel3t){\Assign{\var{reg0}}{\var{y}};} ; \\
            \myindent{1}\tikz[overlay, remember picture]\node[anchor=west](Irel4t){\Assign{\var{x}}{\var{reg0}}};
        }} \\ \ldots$}
    \end{subfigure}
    \hfill
    \begin{subfigure}[t]{0.45\textwidth}
        {\footnotesize$ \ldots \\
        \vspace{-0.33\baselineskip}
        \ITEelse{
            \myindent{1}\tikz[overlay, remember picture]\node[anchor=west](Irel1e){\Assign{\var{reg1}}{\var{y}};} ; \\
            \myindent{1}\tikz[overlay, remember picture]\node[anchor=west](Irel2e){\Assign{\var{reg2}}{\var{z}};} ; \\
            \myindent{1}\tikz[overlay, remember picture]\node[anchor=west](Irel3e){\Assign{\var{reg0}}{\var{reg1} + \var{reg2}};} ; \\
            \myindent{1}\tikz[overlay, remember picture]\node[anchor=west](Irel4e){\Assign{\var{x}}{\var{reg0}}};
        }$}
        \hfill
    \end{subfigure}
    \caption{Concrete if-conditional.
        Relation $\mathcal{I}$ pairs configurations of this program as shown by the dashed lines.}
        %The solid red lines are examples of pairs excluded by $\mathcal{I}$.
        \label{fig:refinement-example-conc}
        \begin{tikzpicture}[overlay, remember picture]
            \draw[dashed] (Irel1t)--(Irel1e);
            \draw[dashed] (Irel2t)--(Irel2e);
            \draw[dashed] (Irel3t)--(Irel3e);
            \draw[dashed] (Irel4t)--(Irel4e);
            %\draw[thick,red] (Irel4t)--(Irel1e);
            %\draw[thick,red] (Irel1t)--(Irel3e);
        \end{tikzpicture}
    \end{subfigure}
    %\vspace{-1\baselineskip}
    \caption{Excerpts from refinement example \cite{Murray_SPR_16} that was used to compare proof effort (\autoref{sec-h-branch-example}).}
    \label{fig:refinement-example}
\end{figure}

\subsection{CVDNI-preserving refinement}\label{sec:cvdni-refinement}

Having described the formal security property that we wish to be preserved under refinement (and compilation), we now define formally a suitable notion of secure refinement that preserves it.
The proof of CVDNI-preserving refinement for a thread of a concurrent program relies on two binary relations (illustrated by \autoref{fig:refinement-example}) to be nominated by the user of the theory:
\begin{enumerate}
    \item a \emph{refinement relation} $\mathcal{R}$ relating local configurations of the abstract program to local configurations of the concrete program:
        abstract must simulate concrete, in a sense typical of much other work on program refinement, including compiler verification efforts.
    \item a \emph{concrete coupling invariant} $\mathcal{I}$ that allows us to use $\mathcal{B}$ and $\mathcal{R}$ to build a new strong low-bisimulation (modulo modes) for the concrete program, by discarding unreachable pairs of local configurations \emph{after the refinement}.
        It thereby witnesses that any changes a refinement (or compiler) makes to execution time, do not introduce any timing channels.
\end{enumerate}

The essence of the proof technique is to require that a number of conditions---analogous to those for $\StrongLowBisimMM$---be imposed on the nominated $\mathcal{R}$ and $\mathcal{I}$ in relation to a given witness relation $\mathcal{B}$ establishing CVDNI for the abstract program.
The definitions to follow are adapted from Murray et al.~\cite{Murray_SPR_16} Section V.
For better readability, we present a simplified version in which no new shared variables are added by the refinement.
Consequently we introduce the notation $\LCSameMdsMemOp$ to denote that two local configurations have equal mode state and memory, regardless of whether relating configurations of the same or differing languages.

Regarding the maintenance of modes- and observational-equivalence across the relation, the restrictions on refinement are tighter than those that applied to $\StrongLowBisimMM$.
The refinement relation $\mathcal{R}$ is required to preserve the shared memory in its entirety:
\begin{definition} [Preservation of modes and memory] \label{def:preserves_modes_mem}
    \[
      \PreservesMM\ \mathcal{R}\ \equiv \ 
      \forall \lc_{A}\ \lc_{C}. \ 
      (\lc_{A}, \lc_{C}) \in \mathcal{R} \longrightarrow \ 
      \LCSameMdsMem{\lc_A}{\lc_C}
    \]
\end{definition}

Regarding the closedness under changes by other threads that ensures compositionality for concurrency, on $\mathcal{I}$ we again impose
$\CgConsistent$ (\autoref{def:cg-consistent}) from \autoref{sec:cvdni-security}.
However in the case of $\mathcal{R}$, we instead impose
$\ClosedOthers$, a simplification of $\CgConsistent$ considering only environmental actions that affect the memories on both sides of the relation identically.
Furthermore it ensures equality of \emph{all} shared variables, not just those judged observable:
\begin{definition} [Closedness of refinements under changes by others] \label{def:closed-others}
\begin{align*}
    & \ClosedOthers\ \mathcal{R}\ \equiv \ 
      \forall \tps_A\ \tps_C\ \mds\ \mem\ \mem'. \\
    & \qquad (\LocalConfAbs{\tps_A}{\mds}{\mem},
              \LocalConfConc{\tps_C}{\mds}{\mem}) \in \mathcal{R})\ \land \\
    & \qquad (\forall x.\ (\mem\ x \ne mem'\ x\ \lor \ 
                           \dmaApp{\mem}{x} \ne \dmaApp{\mem'}{x}) \ 
                        \longrightarrow \ 
                        \Writable\ \mds\ x)\ \longrightarrow \\
    & \qquad (\LocalConfAbs{\tps_A}{\mds}{\mem'},
              \LocalConfConc{\tps_C}{\mds}{\mem'}) \in \mathcal{R})
\end{align*}
\end{definition}

The final major requirement for CVDNI-preservation is then to prove $\mathcal{R}$ and $\mathcal{I}$ closed simultaneously under the pairwise executions of the concrete and abstract programs, using the aforementioned cube-shaped diagram ($\CouplingInvPres$, \autoref{fig:coupling-inv-pres}) whose edges are pairs in $\mathcal{B}$, $\mathcal{R}$, and $\mathcal{I}$.
All that then remains is for the nominated concrete coupling invariant $\mathcal{I}$ to be symmetric, and the predicate $\SecureRefinement$ puts together all the requirements:
\begin{definition} [Requirements for secure refinement of the per-thread CVDNI property]\label{def:secure-refinement}
\begin{align*}
    \SecureRefinement\ \mathcal{B}\ \mathcal{R}\ \mathcal{I}\ \equiv\ 
  & \PreservesMM\ \mathcal{R}\ \land \ 
    \ClosedOthers\ \mathcal{R}\ \land \\
  & \CgConsistent\ \mathcal{I}\ \land \ 
    \Sym\ \mathcal{I}\ \land \ 
    \CouplingInvPres\ \mathcal{B}\ \mathcal{R}\ \mathcal{I}
\end{align*}
\end{definition}
Theorem 5.1 of our prior work \cite{Murray_SPR_16} gives us that under the aforementioned conditions,
\begin{align*}
\BCofApplied \equiv
\{(\lc_{1C}, \lc_{2C})\ |\ 
    & \exists \lc_{1A}\ \lc_{2A}.\ 
      (\lc_{1A}, \lc_{1C}) \in \mathcal{R}\ \land \ 
      (\lc_{2A}, \lc_{2C}) \in \mathcal{R}\ \land \\
    & (\lc_{1A}, \lc_{2A}) \in \mathcal{B}\ \land \ 
                  \LCLowMdsEq{\lc_{1C}}{\lc_{2C}}\ \land \ 
                  (\lc_{1C}, \lc_{2C}) \in \mathcal{I}\}
\end{align*}
is a witness $\StrongLowBisimMM$ for the concrete program:
\[
    \StrongLowBisimMM\ \mathcal{B}\ \land\ \SecureRefinement\ \mathcal{B}\ \mathcal{R}\ \mathcal{I} \implies \StrongLowBisimMM\ (\BCofApplied)
\]

\section{Decomposition principle for CVDNI-preserving refinement}\label{sec-simpler}

Having presented our previous work \cite{Murray_SPR_16}'s formalisation of our security property CVDNI and its preservation by refinement, we now present our first contribution:
an alternative way of proving $\SecureRefinement$ (\autoref{def:secure-refinement})
that does away with the use of the cube-shaped, two-sided refinement obligation $\CouplingInvPres\ \mathcal{B}\ \mathcal{R}\ \mathcal{I}$ (depicted by \autoref{fig:coupling-inv-pres}), by decomposing its concerns into (1) proving $\mathcal{R}$ closed under the pairwise executions of the concrete and abstract programs alone using a square-shaped diagram (depicted by \autoref{fig:decomp-R}, which is akin to ordinary semantics-preserving refinement), and (2) a number of smaller and more separable obligations gathered together under the side-condition predicate $\SimplerRefinementSafe$.
\begin{definition} [Decomposed requirements for CVDNI-preserving secure refinement]\label{def:secure-refinement-simpler}
\begin{align*}
    & \SecureRefineSimpler\ \mathcal{B}\ \mathcal{R}\ \mathcal{I}\ \abssteps\ \equiv \\
    & \PreservesMM\ \mathcal{R}\ \land \ 
      \ClosedOthers\ \mathcal{R}\ \land \ 
      \CgConsistent\ \mathcal{I}\ \land \ 
      \Sym\ \mathcal{I}\ \land \\
    & \SimplerRefinementSafe\ \mathcal{B}\ \mathcal{R}\ \mathcal{I}\ \abssteps\ \land \ 
      (\forall \lc_A\ \lc_C.\ (\lc_A, \lc_C) \in \mathcal{R} \longrightarrow \\
    & \qquad (\forall \lc_C'.\ \lc_C \EvalStepConc \lc_C' \longrightarrow \ 
             (\exists \lc_A'.\ \lc_A \NEvalStepAbs{(\abssteps\ \lc_A\ \lc_C)} \lc_A'\ \land \ (\lc_A', \lc_C') \in \mathcal{R})))
\end{align*}
\end{definition}
The decomposition requires the provision of a new refinement parameter that we will call $\abssteps$ or the \emph{pacing function}, whose role is to dictate the pace of the refinement by returning the number of abstract steps that ought to be taken for a single concrete step, for a given abstract-concrete local configuration pair related by $\mathcal{R}$.
The side-conditions on all of the refinement parameters (depicted by Figures \ref{fig:decomp-abs-steps}, \ref{fig:decomp-I}) are then defined as follows:

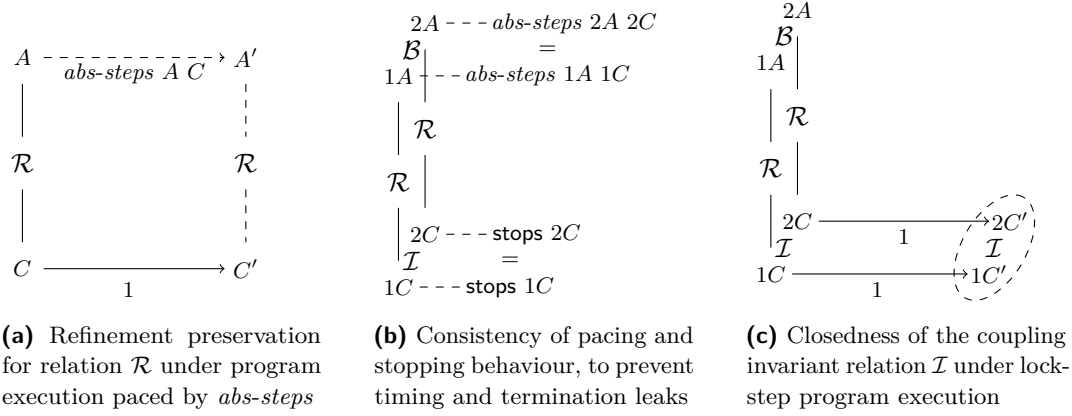
\begin{figure}
    \centering
    % All three subfigures are adapted by Robert Sison from
    % coupling_invariant.tex from our CSF'16 paper repo,
    % originally by Toby Murray. -robs.
    \begin{subfigure}[b]{0.3\textwidth}
        \begin{tikzpicture}[scale=0.7]
            \node at (0,0) {\footnotesize{$A$}};
            \draw [->, dashed] (0.4,0) -- (3.8,0);;
            \node at (2.1,-0.3) {\footnotesize{$\abssteps\ A\ C$}};
            \node at (4.2,0) {\footnotesize{$A'$}};

            \draw (0,-0.5) -- (0,-1.5);
            \node at (0,-2) {$\mathcal{R}$};
            \draw (0,-2.5) -- (0,-3.5);

            \draw [dashed] (4.2,-0.5) -- (4.2,-1.5);
            \node at (4.2,-2) {$\mathcal{R}$};
            \draw [dashed] (4.2,-2.5) -- (4.2,-3.5);

            \node at (0,-4) {\footnotesize{$C$}};
            \draw [->] (0.4,-4) -- (3.8,-4);;
            \node at (2,-4.4) {\footnotesize$1$};
            \node at (4.2,-4) {\footnotesize{$C'$}};
        \end{tikzpicture}
        \caption{Refinement preservation for relation $\mathcal{R}$ under program execution paced by $\abssteps$}\label{fig:decomp-R}
    \end{subfigure}\hfill
    \begin{subfigure}[b]{0.3\textwidth}
        \begin{tikzpicture}[scale=0.7]
            \node at (0,0) {\footnotesize{$1A$}};
            \draw [-, dashed] (0.4,0) -- (1.2,0);;
            \node at (2.8,0) {\footnotesize{$\abssteps\ 1A\ 1C$}};
            \node at (2.8,0.5) {\footnotesize{=}};

            \node at (0.5,1) {\footnotesize{$2A$}};
            \draw [-, dashed] (0.9,1) -- (1.7,1);;
            \node at (3.3,1) {\footnotesize{$\abssteps\ 2A\ 2C$}};

            \node at (0.25,0.5) {$\mathcal{B}$};

            \node at (0.25,-3.5) {$\mathcal{I}$};

            \draw (0,-0.5) -- (0,-1.5);
            \node at (0,-2) {$\mathcal{R}$};
            \draw (0,-2.5) -- (0,-3.5);

            \draw (0.5,0.5) -- (0.5,-0.5);
            \node at (0.5,-1) {$\mathcal{R}$};
            \draw (0.5,-1.5) -- (0.5,-2.5);

            \node at (0,-4) {\footnotesize{$1C$}};
            \draw [-, dashed] (0.4,-4) -- (1.2,-4);;
            \node at (2.1,-4) {\footnotesize{\textsf{stops} $1C$}};
            \node at (2.1,-3.5) {\footnotesize{=}};

            \node at (0.5,-3) {\footnotesize{$2C$}};
            \draw [-, dashed] (0.9,-3) -- (1.7,-3);;
            \node at (2.6,-3) {\footnotesize{\textsf{stops} $2C$}};
        \end{tikzpicture}
        \caption{Consistency of pacing and stopping behaviour, to prevent timing and termination leaks}\label{fig:decomp-abs-steps}
    \end{subfigure}\hfill
    \begin{subfigure}[b]{0.3\textwidth}
        \begin{tikzpicture}[scale=0.7]
            \node at (0,0) {\footnotesize{$1A$}};

            \node at (0.5,1) {\footnotesize{$2A$}};

            \node at (0.25,0.5) {$\mathcal{B}$};

            \node at (0.25,-3.5) {$\mathcal{I}$};
            \node at (4.2,-3.5) {$\mathcal{I}$};
            \draw [dashed] (4.2,-3.5) circle [x radius=1.1, y radius=0.6, rotate=60];;

            \draw (0,-0.5) -- (0,-1.5);
            \node at (0,-2) {$\mathcal{R}$};
            \draw (0,-2.5) -- (0,-3.5);

            \draw (0.5,0.5) -- (0.5,-0.5);
            \node at (0.5,-1) {$\mathcal{R}$};
            \draw (0.5,-1.5) -- (0.5,-2.5);

            \node at (0,-4) {\footnotesize{$1C$}};
            \draw [->] (0.4,-4) -- (3.7,-4);;
            \node at (2,-4.3) {\footnotesize$1$};
            \node at (4.1,-4) {\footnotesize{$1C'$}};

            \node at (0.5,-3) {\footnotesize{$2C$}};
            \draw [->] (0.9,-3) -- (4.1,-3);;
            \node at (2.5,-3.3) {\footnotesize$1$};
            \node at (4.5,-3) {\footnotesize{$2C'$}};
        \end{tikzpicture}
        \caption{Closedness of the coupling invariant relation $\mathcal{I}$ under lockstep program execution}\label{fig:decomp-I}
    \end{subfigure}
    \caption{Graphical depictions of refinement decomposition obligations\label{fig:decomp}}
\end{figure}

\begin{definition} [Side-conditions for CVDNI-preserving refinement decomposition] \label{def:simpler-refinement-safe}
\begin{align*}
    & \SimplerRefinementSafe\ \mathcal{B}\ \mathcal{R}\ \mathcal{I}\ \abssteps\ \equiv \ 
      \forall \lc_{1A}\ \lc_{2A}\ \lc_{1C}\ \lc_{2C}. \ 
      (\lc_{1A}, \lc_{2A}) \in \mathcal{B}\ \land \\
    & \LCSameMds{\lc_{1A}}{\lc_{2A}}\,\land
      (\lc_{1A}, \lc_{1C}) \in \mathcal{R}\,\land
      (\lc_{2A}, \lc_{2C}) \in \mathcal{R}\,\land
      (\lc_{1C}, \lc_{2C}) \in \mathcal{I}\land\,
      \LCSameMds{\lc_{1C}}{\lc_{2C}} \\
    & \longrightarrow\,\stops\ \lc_{1C} = \stops\ \lc_{2C}\ \land \ 
             \abssteps\ \lc_{1A}\ \lc_{1C} = \abssteps\ \lc_{2A}\ \lc_{2C}\ \land \\
    & \quad\ \  (\forall \lc_{1C}'\ \lc_{2C}'. \ 
             \lc_{1C} \EvalStepConc \lc_{1C}'\ \land \ 
             \lc_{2C} \EvalStepConc \lc_{2C}' \longrightarrow \ 
             (\lc_{1C}', \lc_{2C}') \in \mathcal{I}\,\land \,
             \LCSameMds{\lc_{1C}'}{\lc_{2C}'})
\end{align*}
\end{definition}
On the intuitive meaning of the side-conditions in \autoref{def:simpler-refinement-safe}:
\begin{itemize}
    \item $\stops\ \lc_{1C} = \stops\ \lc_{2C}$ ensures that the refinement has not introduced any termination leaks, by asserting \emph{consistent stopping behaviour} for $\mathcal{I}$-related concrete program configurations, which we know to be observationally indistinguishable.
    \item $\abssteps\ \lc_{1A}\ \lc_{1C} = \abssteps\ \lc_{2A}\ \lc_{2C}$ ensures that the refinement has not introduced any timing leaks, by asserting \emph{consistency of the pace of the refinement} for $\mathcal{R}$-related program configurations, which we again know to be observationally indistinguishable.
    \item The final $\forall$-quantified clause asserts $\mathcal{I}$'s suitability as a coupling invariant, in that it must remain \emph{closed under lockstep evaluation} of the concrete program configurations it relates.
        Furthermore it must \emph{maintain mode state equality} with each lockstep evaluation, which ensures that the refinement has not introduced any inconsistencies in the memory access assumptions and guarantees needed for the concurrent compositionality of the property.
\end{itemize}
Note the $\mathcal{B}$- and $\mathcal{R}$-edges in \autoref{fig:decomp-I} may capture useful
facts about a particular program verification technique and compiler, so their availability as assumptions is intended to reduce greatly the effort needed to specify a coupling invariant $\mathcal{I}$ and prove it satisfies the condition.

Assuming the fulfilment of all of the decomposed requirements, we obtain that they are a sound method for establishing secure refinement of the per-thread CVDNI property:
\begin{theorem} [Soundness of $\SecureRefineSimpler$] \label{thm:secure-refinement-simpler-sound}
\begin{align*}
\SecureRefineSimpler\ \mathcal{B}\ \mathcal{R}\ \mathcal{I}\ \abssteps\ \implies \SecureRefinement\ \mathcal{B}\ \mathcal{R}\ \mathcal{I}
\end{align*}
\end{theorem}
In the interests of brevity we relegate proof sketches for all results to
\ifPreprint
Appendices \ref{sec:elided-decomp} and \ref{sec:elided-compiler-section},
\else
the extended version of the paper,
\fi
and for fuller details we refer the reader to our Isabelle/HOL formalisation.

We now devote our attention to two instantiations of this new decomposition principle: (\autoref{sec-h-branch-example}) for a proof of CVDNI-preservation for the refinement of a program that branches on a secret, and (\autoref{sec-wr-compiler-proof}) for the proof of CVDNI-preservation by a compiler.

\section{Proof effort comparison}\label{sec-h-branch-example}

To demonstrate how the decomposition principle reduces proof complexity and effort, we returned to the example refinement discussed in Section V-E of our previous work \cite{Murray_SPR_16}, an excerpt of which is shown in \autoref{fig:refinement-example}.
The abstract program (9 imperative commands) branches on a sensitive value, and executes a single atomic expression assignment in each branch.
Its refinement (to 16 commands) models expansion of the expressions into multiple steps, resolving a timing disparity between the two branches by padding with $\Skip$.

We use proof size as a proxy for proof effort, since the former
is known to be strongly linearly correlated with the latter~\cite{Staples_JAMKK_14}.
Formalised in Isabelle/HOL as \texttt{EgHighBranchRevC.thy} \cite{Murray16-Refinement-AFP}, the proof line count for that theory stood at about 4.6K lines of definitions
and proof, of which approx.\ 3.6K line were proofs.
Adapting the proof instead to use the decomposition principle $\SecureRefineSimpler$ (\autoref{def:secure-refinement-simpler}),
the proof line count drops from 3.6K to approx.\ 2K, a 44\% reduction.
Regarding definition changes, the new proof makes <10 lines of adaptations to a coupling invariant and pacing function used by the old proof, and adds about 30 lines worth of new helper definitions, for use with the decomposition principle.
The rest of the theory and its external dependencies remain in common between the two versions.

As would be expected, the bulk of the deletions are from the full cube-shaped refinement diagram proof (\autoref{fig:coupling-inv-pres}) of $\SecureRefinement$ (\autoref{def:secure-refinement}) for the refinement relation.
The surviving parts of that proof just become the square-shaped refinement diagram proof (\autoref{fig:decomp-R}) of $\SecureRefineSimpler$ without much modification.
The deletions are replaced by newly added proofs of the three sub-obligations of $\SimplerRefinementSafe$ (\autoref{def:simpler-refinement-safe}).

\section{The {\mdseries\Covern} \WRCompiler}\label{sec-compiler}

Having presented our new decomposition principle for CVDNI-preserving refinement, we now turn to our compiler, whose most notable features for formal proof of secure refinement are:
\begin{enumerate}
    \item Its implementation tracks variable stability (\autoref{sec:comprec}) responsive to use of locking primitives, to know when accesses to shared variables are safe to optimise, and when register contents can be still be considered consistent with shared variable contents.
    \item Its verification uses a pacing function (\autoref{sec:abs-steps_wr}) and coupling invariant (\autoref{sec-wr-coupling-invariant}) as the decomposition demands, to ensure it does not introduce timing leaks.
\end{enumerate}
First, we describe its source and target languages, and parameters to the compilation.

\subsection{Source language}

The \Covern \WRCompiler---short for \emph{While-to-RISC compiler}---takes the simple imperative language with while-looping and lock-based synchronisation targeted by the \Covern program logic \cite{Murray_SE_18}, which we will refer to as \texttt{While}, consisting of the commands $cmd$:
\[
\begin{array}{r@{\ }l}
\expr \deftype & n\ |\ v\ |\ \expr\ \oplus \expr \\
\cmd \deftype & \Skip\ |\ \Seqg{\cmd}{\cmd}\ |\ \ITEg{exp}{\cmd}{\cmd}\ |\ \\
         & \Whileg{exp}{\cmd}\ |\ \Assign{v}{exp}\ |\ \\
         & \LockAcq{k}\ |\ \LockRel{k}
\end{array}
\]

The language is parameterised over a type of values~$\Val$, and binary operators $\oplus \oftype \Val \Rightarrow \Val \Rightarrow \Val$. Constants~$n \oftype \Val$;  $v \oftype \Var$ and $k \oftype \Lock$ are (resp.) shared program- and lock-variables.
The semantics of the locking primitives $\LockAcq{k}$ and $\LockRel{k}$ is informed by a locking discipline provided by the user of the theory as a parameter (see \autoref{sec:dma-lock-interp-reqs}).
We leave for future work adding support for pointers and arrays, which we believe will be straightforward because our assume-guarantee framework already provides the means to encode the memory footprint of a command in a way that depends on values in memory.

We assume that the underlying concurrent execution model (e.g.\ operating system, scheduler) for the \texttt{While} language prevents threads from seeing each others' current program location, and thus (as in previous work \cite{Murray_SPR_16,Mantel_SS_11}) the \texttt{While} program command $c \oftype \cmd$ being executed we model as thread-private state: $\LocalConfWhile{c}{\mds}{\mem}$.
In contrast, all program variables $v \oftype \Var$ and lock variables $k \oftype \Lock$ reside in the shared memory $\mem$.

\subsection{Target language}\label{sec:target-lang}

The \WRCompiler's target is a generic RISC-style assembly language like that of Tedesco et al.~\cite{Tedesco16} but with lock-based synchronisation primitives added, which we will refer to as \texttt{RISC}:
\[
\begin{array}{r@{\ }l}
I \deftype & [l :] B \\
B \deftype & \Load\ r\ v\ |\ \Store\ v\ r\ |\ \Jmp\ l\ |\ \Jz\ l\ r\ |\ \Nop \\
      & \MoveK\ r\ n\ |\ \MoveR\ r\ r\ |\ \Op\ \oplus\ r\ r \\
      & \RISCLockAcq\ k\ |\ \RISCLockRel\ k
\end{array}
\]

The language is parameterised over the same value type $Val$ and binary operators $\oplus$, shared program variables $v \oftype \Var$ and shared lock variables $k \oftype \Lock$ as the \texttt{While} language.
Presently, direct-addressing $\Load$ and $\Store$ instructions (referring to registers $r \oftype \Reg$) are adequate for \texttt{RISC} to implement all existing \texttt{While} features, and we expect adding indirect addressing to \texttt{RISC} to be as straightforward as adding pointer and array support to \texttt{While}.

\texttt{RISC} program texts $P$
are just lists of binary instructions $I$, each optionally associated with a label $l \oftype \Lab$.
We assume that the underlying concurrency model for the \texttt{RISC} language (e.g. OS, scheduler etc.) prevents one thread from reading the program code (instructions) of another,\footnote{As is usual for program analyses, we omit any explicit modelling of the microarchitectural state used by superscalar processors (like CPU caches, and state relied on by speculative and out-of-order execution, on whose behaviour attacks like Spectre \cite{Kocher2018spectre} and Meltdown \cite{Lipp2018meltdown} relied).
We argue however that our present assumptions are reasonable under two circumstances: when there is no such state (e.g. on microcontrollers like AVR \cite{Dewald17}),
or when such state is correctly \emph{partitioned} by the underlying hardware \cite{Zhang15} or the OS \cite{Ge_YCH_19} -- if the hardware allows it \cite{Ge_YH_18}!
In the latter case, our analysis assumes that microarchitectural state footprints are partitioned according to thread (for memory containing program text) and according to classification by $\dmaFunc$ (for shared memory), and furthermore that each value-dependently classified region is given a distinct partition that is flushed on reclassification.}
as well as another's registers (including the program counter).
Thus, we model the distinguished program counter register's value $\pc \oftype \nat$, program text $P$, and register bank $\regs \oftype \Reg \Rightarrow \Val$ as thread-private state:
$\LocalConfRISC{\pc}{P}{\regs}{\mds}{\mem}$.
Apart from this adaptation to our triple format, evaluation semantics follows that of the \texttt{RISC} target of \cite{Tedesco16}.

% This became a stub, so we'll roll it into target language -robs.
% Details were moved to the appendix.
%\subsection{Register allocation scheme} \label{sec:reg-alloc}

Finally, like Tedesco et al.~\cite{Tedesco16} we generalise over the (user-supplied) register allocation scheme, and assume there are enough registers to service the maximum depth of expressions in the source program.
\ifPreprint
(More details are available in Appendix \ref{sec:reg-alloc}.)
\fi
We leave for future work the modelling and analysis of a compiler phase that spills register contents to memory, in order to make this assumption unnecessary.

\subsection{Locking discipline}\label{sec:dma-lock-interp-reqs}

% Point 1 -robs.
Like the \Covern logic~\cite{Murray_SE_18}, we assume that the \texttt{While} language program being compiled follows a certain locking discipline, about which the compiler has knowledge, so as to ensure that the \texttt{RISC} program it produces follows the same discipline.

The user of the theory provides the details of the locking discipline in the form of a \emph{lock interpretation} parameter:
$\lockinterp \oftype \Lock \Rightarrow (\Var\ \set \times \Var\ \set)$,
% robs: I must word it the following way because if nobody is holding lock l
% which is responsible for NoW for variable x, then they all guarantee not
% to write to x. However, they may all still read from x.
which for each lock gives the two non-overlapping sets of program variables over which acquiring the lock grants exclusive permission to write, (resp.) read and write.
% Point 2 -robs.
These permissions are then reflected in the way the semantics of the \texttt{While} and \texttt{RISC} locking primitives act on the mode state.

Regarding lock interpretations and the way they interact with the user-provided value-dependent classification function $\dmaFunc$ (see \autoref{sec:cvdni-security}), we inherit a few cleanliness conditions from that earlier work~\cite{Murray_SE_18}, chief of which are that lock variables $k$ cannot be control variables, a lock variable $k$ governing access to a program variable $v$ must govern the same kind of access to all of $v$'s control variables, and $\dmaFunc$ must classify all lock variables as $\Low$.

\subsection{Compiler implementation and tracking of shared variable stability}\label{sec:comprec}

We chose as a starting point the
compilation scheme of \cite{Tedesco16}, on the basis of their preserving a noninterference property that like ours exhibits resilience to changes made by an environment---in their case, intended for fault-resilience.
Aiming to repurpose that for shared-variable concurrency, we adapted it to Isabelle, implementing it as a primitive recursive function:
\begin{align*}
    \CompileCmd \oftype\ 
    & CompRec \Rightarrow \Lab\ option \Rightarrow \Lab \Rightarrow cmd \Rightarrow \\
    & (I \times CompRec)\ list \times \Lab\ option \times \Lab \times CompRec \times bool
\end{align*}
where we choose $\Lab=\nat$ for \texttt{RISC} instruction labels, and the \emph{compilation record} type $CompRec$ is bookkeeping maintained by the compiler that we will describe further below.

A typical invocation to compile a \texttt{While} program $c \oftype cmd$ takes the form:
\begin{equation}\label{eqn:example-invocation}
  (\PCs, l', \nl', C', \failed) = \CompileCmd\ C\ l\ \nl\ c
\end{equation}
Here, $\CompileCmd$ takes an \emph{initial compilation record} $C$, an optional \emph{entry label} $l$, and the \emph{next available label} $\nl$, and for the benefit of the next invocation returns an optional \emph{exit label} $l'$ if one is used by the program just compiled, the \emph{new next available label} $\nl'$, and a \emph{final compilation record} $C'$.
We leave details of label allocation and its impact on achieving sequential composability for compiled \texttt{RISC} programs to
\ifPreprint
Appendix \ref{sec:labels-seq}.
\else
the extended version of the paper.
\fi

In addition to the output \texttt{RISC} program $P \oftype I\ list$ itself, a call to $\CompileCmd$ also outputs every $CompRec$ associated with the state of the program just before executing every instruction in $P$.
These are returned zipped up together with $P$ as the \emph{$CompRec$-annotated \texttt{RISC} program} $\PCs \oftype (I \times CompRec)\ list$.
($P$ can trivially be recovered as $\mapfst\ \PCs$.)
Finally, $\CompileCmd$ may return $\True$ for $\failed$ to reject the input program,
such as when it detects a data race (see below),
or if expression depth exceeds the assumed limit (\autoref{sec:target-lang}).

In the style of the compilation scheme on which it was based \cite{Tedesco16}, the \WRCompiler maintains a \emph{register record} $\Phi \oftype reg \rightharpoonup exp$, i.e. a partial map of registers to expressions on shared variables.
In addition to using it to compile away any unnecessary loads from variables in shared memory, we also use it
to ensure that an expression calculated by \texttt{RISC} in registers
is equal to the value of the expression as if it had all been calculated by \texttt{While} in one step.
This is especially important when writing the result of an expression back to shared memory, because the refinement is required to maintain all shared memory values.

New to the \WRCompiler is the responsibility of maintaining an \emph{assumption record}, which it uses primarily to detect and reject programs with data races on shared memory, and to rule out the introduction of any new ones.
Each assumption record $\mathcal{S} \oftype (\Var\ \set \times \Var\ \set)$ is a pair tracking the set of variables on which (resp.) \textbf{AsmNoW}, \textbf{AsmNoRW} assumptions are currently active at a given point in the program being compiled.
As a secondary concern we also use it to assert that the two sides of any if-conditional branches act consistently on the mode state, and that while-loops restore the original mode state on termination.

A compilation record $C = (\Phi, \mathcal{S}) \oftype CompRec$ is then just a register/assumption record pair.
For readability, we use $\Regrec$, $\Asmrec$ to denote (resp.) a $CompRec$'s $\fst$, $\snd$ projections.

To explain how the compilation record is used to rule out data races, and to ensure consistency of expression evaluation between source and target program, firstly we must introduce the concept of \emph{stability} of a variable $v$ according to an assumption record $\mathcal{S}$:
\[
    \VarStable\ \mathcal{S}\ v\ \equiv \ 
    v \in (\fst\ \mathcal{S} \cup \snd\ \mathcal{S})\ \land \ 
    (\forall v' \in \Cvars\ v.\ 
    v' \in (\fst\ \mathcal{S} \cup \snd\ \mathcal{S}))
\]
In short, this means that the variable and all its control variables ($\Cvars\ v$) are recorded as having either of \textbf{AsmNoW} or \textbf{AsmNoRW} active on them.

For register record entries to be of any help in ensuring consistency of \texttt{While} and \texttt{RISC} expression evaluation, we exclude expression evaluation on data race-prone variables by lifting the concept of stability to register records.
The following predicate asserts internal consistency of the compilation record $C$ created by $\CompileCmd$, in the sense that the register record may only map to expressions that mention variables that are recorded as \textsf{stable} by the assumption record accompanying it.
(Here, $\defined{ran}$ denotes the \emph{range} of a map.)
\[
    \RegrecStable\ C\ \equiv \ 
    \forall e \in \defined{ran}\ (\Regrec\ C). \ 
    (\forall v \in \mathsf{exp{\text -}vars}\ e.\ 
    \VarStable\ (\Asmrec\ C)\ v)
\]

To ensure that an input \texttt{While} program maintains register record stability, we define the predicate \textsf{no-unstable-exprs} $c$ $C$ to capture the requirement that a program~$c$, if started with a configuration consistent with compilation record $C$, will never access a lock-protected variable without holding the relevant lock.
(It also checks the secondary, mode-state consistency concerns of the assumption record mentioned earlier.)
We implement it as a simple static check carried out by a primitive recursive function on the structure of \texttt{While} programs.

Together, $\RegrecStable$ and \textsf{no-unstable-exprs} make up the main two requirements of a predicate $\CompilerInputReqs\ C\ l\ \nl\ c$ imposed on the input arguments to $\CompileCmd$,
which gives us enough information to prove a lemma that $\CompileCmd$ only ever outputs stable register records.
Full details of these we leave to
\ifPreprint
Appendix \ref{sec:elided-compiler-proofs}.
\else
the extended version of the paper.
\fi

\subsection{Proof of CVDNI-preserving compilation}\label{sec-wr-compiler-proof}

Having covered the most significant aspects of the \Covern \WRCompiler's parameters and machinery, we can now present the refinement relation $\RefRelWR$ (\autoref{sec:refrel_wr}), pacing function $\AbsStepsWR$ (\autoref{sec:abs-steps_wr}), and coupling invariant $\CouplInvWR$ (\autoref{sec-wr-coupling-invariant}) that we use with our new decomposition principle (of \autoref{sec-simpler}) to prove that it preserves CVDNI (\autoref{sec:sec-comp-theorems}).

\subsubsection{Refinement relation $\RefRelWR$ and its invariants}\label{sec:refrel_wr}

Just like our example $\mathcal{R}$ of \autoref{fig:refinement-example}, $\RefRelWR$ pairs abstract with concrete configurations.

Here, we will focus on $\RefRelWR$'s most notable characteristics for understanding why it is suitable to describe a CVDNI-preserving compilation.\footnote{We provide an informal description of all of the cases, their purpose, and the invariants they maintain, along with a code listing from $\CompileCmd$ relevant to the part that will be presented, in
% Note: We will need to make a copy at some point, and maintain it separately,
% because LIPIcs wants a fully-sanitised .tex file,
% but clearly we will want to keep the extended version. -robs.
\ifPreprint
Appendices \ref{sec:R_wr-informal} and \ref{sec:compile-cmd-if} (respectively).
\else
the extended version of our paper.
\fi
For full details, we refer the reader to the Isabelle formalisation.}
We focus on the case $\mathtt{if\_expr}$ of $\RefRelWR$, which
relates the expression evaluation part of the \texttt{While} program $\ITEg{e}{c_1}{c_2}$,
%(For comparison, a code listing of the relevant case of the implementation of $\CompileCmd$ is provided in Appendix \ref{sec:compile-cmd-if}.)
%It has been adapted slightly to improve the clarity of the presentation.
with the corresponding part (including the conditional jump $\Jz$ after expression evaluation) of the \texttt{RISC} program obtained by running $\CompileCmd$ on it.
(Variables ignored are in gray.)

\begin{example} [Introduction rule for case \texttt{if\_expr} of $\RefRelWR$]
\[
{\footnotesize
% It's really annoying that this space applies within as well as between the
% lines of premises. Spaces only between them would be great. FIXME -robs.
\ifPreprint
% *Just* enough not to push Example 10 onto the next page... -robs.
\mprset{vskip=0.46ex}
\else
\mprset{vskip=0.5ex}
\fi
\inferrule{
    c = \ITEg{e}{c_1}{c_2} \and
    \CompilerInputReqs\ C\ l\ \nl\ c \and \\
    (\PCs, \unused{l'}, \nl_2, \unused{C'}, \False) = \CompileCmd\ C\ l\ \nl\ c \and
    \qquad(P_e, \unused{r}, C_1, \False) = \CompileExpr\ C\ \varnothing\ l\ e \and \\
    % Some manual spacing to get compile-cmd assumptions to line up nicely
    (P_1, \unused{l_1}, \nl_1, \unused{C_2}, \False) = \CompileCmd\ C_1\ \None\ (\Suc\ (\Suc\ \nl))\ c_1 \and
        \quad \pc \leq \length\ P_e\ \ \ \ \ \  \and \\
    \,\, (P_2, \unused{l_2}, \nl_2, \unused{C_3}, \False) = \CompileCmd\ C_1\ (\Some\ \nl)\ \nl_1\ c_2 \and
        \quad C_{\pc} = (\mapsnd\ \PCs\ !\ \pc) \and \\
    \CompiledCmdConfigConsistent\ C_{\pc}\ \regs\ \mds\ \mem \and
    \RegrecStable\ C_{\pc} \and \\
    \forall \mds'\ \mem'\ \regs'. \ 
    \CompiledCmdConfigConsistent\ C_1\ \regs'\ \mds'\ \mem'\, \land \,
    \RegrecStable\ C_1 \\
    \quad\  \longrightarrow ((\LocalConfWhile{c_1}{\mds'}{\mem'}, \LocalConfRISC{0}{\mapfst\ P_1}{\regs'}{\mds'}{\mem'}) \in \RefRelWR\ \land \\
    \qquad \quad (\LocalConfWhile{c_2}{\mds'}{\mem'}, \LocalConfRISC{0}{\mapfst\ P_2}{\regs'}{\mds'}{\mem'}) \in \RefRelWR) 
} {
    (\LocalConfWhile{c}{\mds}{\mem}, \LocalConfRISC{\pc}{\mapfst\ \PCs}{\regs}{\mds}{\mem}) \in \RefRelWR
}}
\]
\end{example}
This is a fairly typical case of $\RefRelWR$ in a number of respects:

Firstly, there is a direct reference to the call to $\CompileCmd$ for the given \texttt{While} program.
Secondly, various guards
($\CompiledCmdConfigConsistent$ introduced below, and
$\RegrecStable$ defined in \autoref{sec:comprec})
are asserted in order to restrict the scope of $\RefRelWR$ only to consider
wellformed local program configurations that line up with the conditions captured by the compilation record.
Thirdly, the inductive references to $\RefRelWR$ for $P_1$ and $P_2$, the branches of the conditional \emph{that have not been reached yet}, are quantified over all configurations that obey the guards $\CompiledCmdConfigConsistent$ and $\RegrecStable$ relative to $C_1$, the initial compilation record for each of the sub-calls to $\CompileCmd$ for those sub-programs.

The guard $\CompiledCmdConfigConsistent$ mentioned
above asserts that the compilation record~$C$ is consistent with the
registers~$\regs$, memory~$\mem$ and mode state~$\mds$. 
\begin{align*}
    & \CompiledCmdConfigConsistent\ C\ \regs\ \mds\ \mem\ \equiv \\
    & \qquad (\forall r\ e.\ (\Regrec\ C)\ r = \Some\ e \longrightarrow regs\ r = \mathsf{ev_{exp}}\ \mem\ e)\ \land \\
    & \qquad \Asmrec\ C = (\mds\ \mathbf{AsmNoW},\ \mds\ \mathbf{AsmNoRW})
\end{align*}
Firstly, for all entries in register record mapping some register $r$ to some expression $e$, the value held in $r$ of the register bank $\regs$ must match the value of $e$ if evaluated under memory $\mem$.
Secondly, the assumption record must consist exactly of the program variables the mode state $\mds$ says have \textbf{AsmNoW}, \textbf{AsmNoRW} on them respectively.

As we will see in \autoref{thm:compile-cmd_correctness_R_wr},
$\CompiledCmdConfigConsistent$ also serves as
\emph{initial configuration requirements} for compiled programs: 
only configurations obeying them may be used to initialise a \texttt{RISC} program compiled by the \WRCompiler with initial compilation record $C$.

With $\RefRelWR$ specified, we then prove the two requirements for $\SecureRefineSimpler$ that pertain to $\RefRelWR$ alone:
$\PreservesMM$ (\autoref{def:preserves_modes_mem}) and $\ClosedOthers$ (\autoref{def:closed-others}).
\begin{lemma} [$\RefRelWR$ preserves modes and memory] \label{thm:preserves-modes-mem-R_wr}
  $\PreservesMM\ \RefRelWR$
\end{lemma}
\begin{lemma} [$\RefRelWR$ is closed under changes by others] \label{thm:closed-others-R_wr}
  $\ClosedOthers\ \RefRelWR$
\end{lemma}

\subsubsection{Refinement pacing function $\AbsStepsWR$} \label{sec:abs-steps_wr}

We now nominate an $\abssteps$ function, determining the pace at which \texttt{While} programs progress in comparison to the \texttt{RISC} programs that they are compiled to by the \WRCompiler.

To assist here and elsewhere, we define a primitive recursive helper \textsf{leftmost-cmd} that given a sequence of ;-separated \texttt{While} commands, strips all but the first: given $\Seqg{c_1}{c_2}$ it returns \textsf{leftmost-cmd} $c_1$, and given any other \texttt{While} program $c$ it returns $c$.

Our pacing function $\AbsStepsWR$ primarily looks at the form of the \texttt{RISC} program instruction about to be executed.
The \texttt{RISC} instructions are divided into three categories:
\begin{itemize}
    \item Instructions output by $\CompileExpr$: $\Load$, $\Op$, and $\MoveK$.
        For these, $\AbsStepsWR$ returns 1 if the \textsf{leftmost-cmd} of the \texttt{While} program is $\Whileg{e}{c}$, to allow it to step to $\ITEg{e}{(\Seqg{c}{\Whileg{e}{c}})}{\Stop}$ concurrently with the first \texttt{RISC} step of the compiled expression itself.
        Otherwise, $\AbsStepsWR$ returns 0 to indicate the \texttt{While} program standing still while the \texttt{RISC} program takes \emph{new} steps to evaluate the expression.
    \item ``Epilogue'' steps: $\Jmp$ and $\Nop$ when used for control flow at the end of a smaller compiled program in the context of a larger one.
        For these, $\AbsStepsWR$ returns 0.
    \item All other \texttt{RISC} instructions are assumed to proceed at a lockstep pace with the \texttt{While} command they were compiled from, and for these $\AbsStepsWR$ returns 1.
\end{itemize}

Having nominated $\AbsStepsWR$ and $\RefRelWR$, we now have the parameters over which we are obliged to prove refinement preservation (\autoref{fig:decomp-R}) as demanded by $\SecureRefineSimpler$ (\autoref{def:secure-refinement-simpler}).
To this end, we prove firstly (elided to
\ifPreprint
Appendix \ref{sec:elided-compiler-proofs})
\else
the extended version)
\fi
that every step of execution of a \texttt{RISC} program produced by the \WRCompiler from a \texttt{While} program, maintains the consistency demanded by $\CompiledCmdConfigConsistent$ between configurations and compilation records.
Also, we must prove a correctness lemma for the expression compiler:
\begin{lemma} \label{thm:compile-expr-correct}
$(\PCs, r, C', \False) = \CompileExpr\ C\ A\ l\ e\ 
\implies (\Regrec\ C')\ r = \Some\ e$
\end{lemma}
Armed with these facts, we can now prove the main refinement preservation result:
\begin{lemma} [$\RefRelWR$ is a refinement paced by $\AbsStepsWR$] \label{thm:R_wr-refinement-abs-steps_wr}
\begin{align*}
    & \forall \lc_w\ \lc_r.\ (\lc_w, \lc_r) \in \RefRelWR \longrightarrow \ 
      (\forall \lc_r'.\ \lc_r \EvalStepRISC \lc_r' \longrightarrow \\
    & \qquad (\exists \lc_w'.\ \lc_w \NEvalStepWhile{(\AbsStepsWR\ \lc_w\ \lc_r)} \lc_w'\ \land \ (\lc_w', \lc_r') \in \RefRelWR))
\end{align*}
\end{lemma}

\subsubsection{Concrete coupling invariant $\CouplInvWR$}\label{sec-wr-coupling-invariant}

The next element needed is the concrete coupling invariant $\CouplInvWR$, which we define as follows:

$\CouplInvWR \equiv
\{(\LocalConfRISC{\pc}{P}{\regs}{\mds}{\mem},
   \LocalConfRISC{\pc'}{P'}{\regs'}{\mds'}{\mem'})\ |\ 
   (\pc, P) = (\pc', P')\}$

In other words, $\CouplInvWR$ asserts that we only need compare local configurations that are at the same location $\pc = \pc'$ of the same \texttt{RISC} program $P = P'$.
When used in concert with a $\NoHighBranching\ \mathcal{B}$ (see \autoref{sec:sec-comp-theorems}), the effect of $\CouplInvWR$ is to ensure that the \WRCompiler has not introduced any \emph{new} branching on sensitive values.

\subsubsection{Successful compilations are CVDNI-preserving refinements}\label{sec:sec-comp-theorems}

We are ready to prove preservation.
First we qualify that we allow only $\StrongLowBisimMM\ \mathcal{B}$ that describe only \texttt{While}-programs with no branching on $\High$-classified values, as follows:
\begin{align*}
    & \NoHighBranching\ \mathcal{B} \equiv \\
    & \forall c\ c'\ \mds\ \mem\ \mem'.\ 
      (\LocalConfWhile{c}{\mds}{\mem},
       \LocalConfWhile{c'}{\mds}{\mem'}) \in \mathcal{B}\ 
       \longrightarrow \ c = c'\ \land \\
    & \qquad (\forall e\ c_1\ c_2.\ 
              \mathsf{leftmost{\text -}cmd}\ c = \ITEg{e}{c_1}{c_2} \longrightarrow \ 
              \mathsf{ev_{exp}}\,\mem\ e = \mathsf{ev_{exp}}\,\mem'\,e)
\end{align*}
That is, it refuses to relate configurations at different program locations.
Furthermore if it is at a conditional branching point, the expression $e$ determining which branch will be taken evaluates to the same boolean value for both configurations' memories.
When imposed on a relation that already ensures $\Low$-equivalent memory modulo modes, this effectively disallows any present or past branching on sensitive values.
Then, for such programs:

\begin{lemma} \label{thm:simpler-refinement-safe_wr}
$\begin{aligned}
\inferrule{
    \StrongLowBisimMM\ \mathcal{B} \and
    \NoHighBranching\ \mathcal{B}
} {
    \SecureRefineSimpler\ \mathcal{B}\ \RefRelWR\ \CouplInvWR\ \AbsStepsWR
}\end{aligned}$
\end{lemma}

From this it follows immediately via \autoref{thm:secure-refinement-simpler-sound} that $\RefRelWR$ with the help of $\CouplInvWR$ describes a CVDNI-preserving refinement for non-$\High$-branching \texttt{While} programs:
\begin{corollary} [$\RefRelWR$ is a CVDNI-preserving refinement for non-High-branching programs] \label{thm:secure-refinement-R_wr}
\[
    \StrongLowBisimMM\ \mathcal{B}\ \land \ 
   \NoHighBranching\ \mathcal{B} \implies
   \SecureRefinement\ \mathcal{B}\ \RefRelWR\ \CouplInvWR
\]
\end{corollary}

Finally, we prove that successful compilation produces a \texttt{RISC} program related by $\RefRelWR$ to its input \texttt{While} program, when started with corresponding and reasonable initial configurations:

\begin{theorem} [Successful compilations are refinements in $\RefRelWR$] \label{thm:compile-cmd_correctness_R_wr}
\[
\mprset{vskip=0.5ex}
\inferrule{
  (\PCs, l', \nl', C', \failed) = \CompileCmd\ C\ l\ \nl\ c \and
    \CompilerInputReqs\ C\ l\ \nl\ c \and \\
  \failed = \False \and
    \CompiledCmdConfigConsistent\ C\ \regs\ \mds\ \mem \and
    P = \mapfst\ \PCs
} {
  (\LocalConfWhile{c}{\mds}{\mem},
    \LocalConfRISC{0}{P}{\regs}{\mds}{\mem}) \in \RefRelWR
}
\]
\end{theorem}

\section{Case study: the \WRCompiler in action} \label{sec:exec-instantiation}

To test the theory, we instantiated it and applied the $\WRCompiler$ to a
\texttt{While}-language model of the Cross Domain Desktop Compositor \cite{Beaumont_MM_16} (CDDC), a non-trivial concurrent program
that facilitates a trusted user's interaction with multiple desktop machines of differing clearance.

The CDDC model to which we applied the compiler is a 2-thread program that
was a precursor to the 3-thread model that was verified using the
\Covern program logic \cite{Murray_SE_18}.\footnote{We leave for future work an adaptation of the refinement theory and $\WRCompiler$ in order to support the \emph{shared data invariants} added by the \Covern logic, required to verify the 3-thread CDDC model.}
Each of the threads of the CDDC program
(together about 150 lines of \texttt{While}) we proved
satisfy the compositional security property
$\ComSecure$ (\autoref{def:com-secure}), using a precursor to
the \Covern logic that yields
CVDNI-witness bisimulations that are non-$\High$-branching.

The resulting compiler is \emph{executable} in Isabelle, meaning
that $\CompileCmd$ can be executed on the \texttt{While} program text
for each of the two threads to obtain their compilations
(together totalling about 250 \texttt{RISC} instructions)
using the Isabelle tactic \textbf{eval}.
The secure compilation theorems (\autoref{sec:sec-comp-theorems}), together with $\StrongLowBisimMM$ preservation and compositionality for $\ComSecure$
(Theorems 5.1, 3.1 of \cite{Murray_SPR_16}, mentioned in \autoref{sec-background})
then allow us to derive that the compiled program
is secure when its threads are run concurrently.

To our knowledge this is
the first proof of source-level information-flow security being carried by a verified compiler to an assembly-level model of a non-trivial concurrent program.

\section{Related work} \label{sec:related-work}

% \subsection{Noninterference-preserving compilation}

The following three works, like ours, focus on compilation preserving a form of noninterference.

Tedesco et al.~\cite{Tedesco16} present a type-directed compilation scheme that preserves
a \emph{fault-resilient} noninterference property. The compilation scheme
of our \WRCompiler was inspired by theirs. Like our $\ComSecure$ CVDNI
security property that \WRCompiler preserves, Tedesco et al.'s security property
is also \emph{strong bisimulation}-based~\cite{Sabelfeld00}.
But where our property accounts (via mode states) for \emph{controlled interference} by other threads, theirs instead quantifies over all possible interference by the environment with the memory contents.
While this simplifies their task of proving that their security property is preserved under compilation---as it need not require the compiler to preserve the contents of memory---it means their security property cannot capture value-dependent noninterference.
In contrast, our \WRCompiler must obey our $\SecureRefinement$ notion's requirement that memory contents are preserved.\footnote{Consequently, we found and fixed a bug in their expression compiler (acknowledged privately) whereby registers in use were incorrectly reallocated.
Expressions like $v + (v + 1)$ were thus compiled incorrectly to programs yielding $(v + 1) + (v + 1)$ instead, causing a violation of memory contents preservation.}

Barthe et al.~\cite{Barthe18} consider the problem of preserving
\emph{cryptographic constant-time policies}, %under compilation.
a class of noninterference properties similar to CVDNI in its
explicit consideration for capturing timing-sensitivity.
Barthe et al.~consider a wider scope of common categories of compile-time optimisations (than those performed by our \WRCompiler), and mechanise proofs in Coq that
such optimisations preserve various constant-time security properties.
The sharing of variables in our setting
severely limits the scope of our optimisations, to those that the compiler can
perform knowing that a shared variable is stable because it has been locked.
At present, our \WRCompiler avoids redundant loads during expression
compilation, but other optimisations like loop hoisting and constant folding we are
yet to implement.
Their preservation proof technique, \emph{constant-time simulation} was developed
independently to our original cube-shaped secure refinement definition~\cite{Murray_SPR_16}. Like ours, theirs is also a cube-shaped obligation and makes use of a pacing
function analogous to our $\abssteps$.
Unlike our work here, Barthe et al.~do not give a general method for decomposing their
cube-shaped simulation diagrams.

% \subsection{Concurrency-compositional noninterference-preserving compilation}

Neither of the above consider per-thread compositional compilation of concurrent, shared memory programs, nor value-dependent noninterference policies -- the focus of our theory and compiler.
Barthe et al.~\cite{Barthe10} however did aim to preserve noninterference of multithreaded programs by compilation, extending a prior \emph{(security) type-preserving} compilation approach~\cite{Barthe07}.
Their noninterference property however was termination- and timing-\emph{insensitive}, so preventing internal timing leaks relied on the scheduler disallowing certain interleavings between threads.
Also, their type-preservation argument was derived from a big-step semantics preservation property for their compiler.
Here we instead rely on preservation of a small-step semantics (specifically memory contents), which is necessary for us to preserve value-dependent security under compilation, as well as to avoid imposing non-standard requirements on the scheduler.

% \subsection{Robust property preservation}

Other recent works have improved on \emph{fully abstract compilation} (surveyed \cite{PatrignaniAC19}) by mapping out the spectrum \cite{Abate18} or developing specific forms \cite{PatrignaniG19} of \emph{robust property preservation}, concerned with \emph{robustness} of source program (hyper)properties to concrete \emph{adversarial} contexts.
Like Tedesco et al.~\cite{Tedesco16}, these works differ from ours in quantifying over a wider range of hostile interference.
They also focus prominently on changes to data types, which we do not support.
Thus, as a 2-safety hyperproperty quantifying over a lesser range of interference, we expect CVDNI-preservation to be implied by R2HSP (robust 2-hypersafety preservation), but do not expect it to imply any other secure compilation criterion on Abate et al.'s \cite{Abate18} spectrum.

While recently Patrignani and Garg \cite{PatrignaniG19} instantiated their \emph{robustly safe compilation} for shared-memory fork-join concurrent programs, it only preserves (1-)safety properties.
Previously however, Patrignani et al.~\cite{Patrignani17} proved their
\emph{trace-preserving compilation} preserves $k$-safety
hyperproperties~\cite{Clarkson10}, including
noninterference properties.
However, it disallows the removal or addition of trace entries, which would be necessary to change the passage of time as seen in the observable trace events.
Thus it excludes optimisations carried out by our compiler (when it permits changes to pacing regulated by $\abssteps$) and studied by the two other works \cite{Tedesco16,Barthe18} on timing-sensitive security-preserving compilation mentioned above.

% \subsection{Compiler verification in general}

Finally, there has been much work on large-scale verified compilation~\cite{Leroy09,Kumar_MNO_14} some of which has also treated compilation of shared-memory concurrent
programs~\cite{Lochbihler18} including taking weak-memory consistency into
account~\cite{Podkopaev19}. Our work here does not consider the effects of weak-memory models.
However, it differs to prior work on verified concurrent compilation, in that it formalises and proves a compiler's ability to use information about the application's locking protocol, to exclude unsafe access to shared variables, and conversely to know when it is safe to allow optimisations that would typically be excluded (see \autoref{sec:comprec}).

\section{Conclusion}

To our knowledge, we have presented the first mechanised verification that a compiler
preserves concurrent, value-dependent noninterference.
To this end, we provided a general
decomposition principle for compositional, secure refinement. Although our compiler is a proof-of-concept targeting simple source and target languages, we nevertheless
applied it to produce a verified assembly-level model of the CDDC \cite{Beaumont_MM_16}, a non-trivial concurrent program.

This work
serves to demonstrate that verified security-preserving compilation for concurrent programs is now
within reach, by augmenting traditional proof obligations for
verified compilation (e.g.\ square-shaped semantics preservation) with those specific
to security (e.g.\ absence of termination- and timing-leaks) as depicted in
\autoref{fig:decomp}. We hope that this work paves the way for future large-scale verified security-preserving compilation efforts.

%%
%% Bibliography
%%

\bibliography{references}

\ifPreprint
\appendix

\section{Informal descriptions of the cases of refinement relation $\RefRelWR$}\label{sec:R_wr-informal}

\subsection{Base cases}

\begin{itemize}
    \item \texttt{stop}: This case relates a terminated \texttt{While} program with a terminated \texttt{RISC} program (i.e. one where the program counter is at the length of the program text).

    \item \texttt{skip\_nop}: This case relates the \texttt{While} program $\Skip$ with the configuration where the program counter is at the start of the \texttt{RISC} program $[\Nop]$.

    \item \texttt{assign\_expr}: This case relates the expression evaluation part (for the expression $e$) of the \texttt{While} program $\Assign{v}{e}$ with the corresponding part of the \texttt{RISC} program obtained by compiling it with the \WRCompiler.

    \item \texttt{assign\_store}: As for \texttt{assign\_expr}, but for the very last $\Store$ instruction that commits the result of the expression evaluation back to shared memory variable $v$.

        It asserts additionally that $v$ must be stable if lock-governed, and non-lock-governed otherwise.
        This prevents threads from violating the locking discipline (see \autoref{sec:dma-lock-interp-reqs}).

    \item \texttt{lock\_acq}: This case relates $\LockAcq{k}$ with $\RISCLockAcq\ k$.
    \item \texttt{lock\_rel}: This case relates $\LockRel{k}$ with $\RISCLockRel\ k$.
\end{itemize}

\subsection{Inductive cases}

\begin{itemize}
    \item \texttt{seq}: This case relates the \texttt{While} program $\Seqg{c_1}{c_2}$ with the \emph{concatenation} $P_1 @ P_2$ of the \texttt{RISC} programs $P_1$ and $P_2$ that are respectively the outputs of successful consecutive compilation of $c_1$ and $c_2$ by the \WRCompiler.
        It is intended for cases where the \texttt{While} (resp. \texttt{RISC}) program is currently in $c_1$ (resp. $P_1$).

        It is an inductive case of $\RefRelWR$, in that:
        \begin{itemize}
            \item $c_1$ is required to be related by $\RefRelWR$ to the present location in $P_1$.
            \item For all local configurations that obey the $\CompiledCmdConfigConsistent$ requirements,
                $c_2$ is required to be related by $\RefRelWR$ to the first instruction of $P_2$.
                This quantification ensures that $\RefRelWR$ remains closed when execution progresses from the first program to the second program.
        \end{itemize}

        It asserts that $P_1$ and $P_2$ are $\Joinable$ (\autoref{sec:labels-seq}), particularly relevant here to ensure that $P_1$ can only jump to locations within or at the end of itself (i.e. the start of $P_2$).

    \item \texttt{join}: This case relates a \texttt{While} program $c$ with an offset $\pc > \length\ P_1$ into a \texttt{RISC} program $P_1 @ P_2$, assuming the inductive hypothesis that $c$ is related by $\RefRelWR$ with the offset $\pc - \length\ P_1$ into the \texttt{RISC} program $P_2$ alone.

        It is intended primarily for cases where the \texttt{While} (resp. \texttt{RISC}) program is currently in the $c_2$ (resp. $P_2$) of some consecutively compiled $\Seqg{c_1}{c_2}$ (resp. $P_1$ concatenated with $P_2$) but applies more broadly to allow any prepend of dead, unreachable instructions onto the front of a \texttt{RISC} program without breaking $\RefRelWR$.

        It also asserts that $P_1$ and $P_2$ are $\Joinable$, which is important here to ensure that $P_2$ cannot jump back into $P_1$.

    \item \texttt{if\_expr}: This case relates the expression evaluation part (for the expression $e$) of the \texttt{While} program $\ITEg{e}{c_1}{c_2}$ with the corresponding part (including the conditional jump $\Jz$ at the end of expression evaluation) of the \texttt{RISC} program obtained by compiling it with the \WRCompiler.

        It relies on both $c_1$ and $c_2$ being related by $\RefRelWR$ to its compiled \texttt{RISC} counterparts when started with initialisation states judged valid by $\CompiledCmdConfigConsistent$.

    \item \texttt{if\_c1}: This case relates some \texttt{While} program $c_1'$ reachable from $c_1$ with the corresponding part within the $c_1$ part of the \texttt{RISC} program obtained by compiling $\ITEg{e}{c_1}{c_2}$ with the \WRCompiler.

        It relies on $c_1$ being related by $\RefRelWR$ to its compiled \texttt{RISC} counterpart at the appropriate program counter offset.

    \item \texttt{if\_c2}: As for \texttt{if\_c1}, but for $c_2$.

    \item \texttt{epilogue\_step}: This case relates a terminated \texttt{While} program to the silent control flow steps navigating to the end of a \texttt{RISC} program from the end of the ``then'' and ``else'' branches of a compiled if-conditional.

        It works only for epilogue step forms (see \autoref{sec:abs-steps_wr}).

        It is inductive in that it asserts closedness of $\RefRelWR$ over pairwise reachability from the pair currently under consideration -- the only case to do so directly.

    \item \texttt{while\_expr}: This case relates the \texttt{While} program ($\Whileg{e}{c}$)'s initial intermediate step to $\ITEg{e}{(\Seqg{c}{\ \Whileg{e}{c}})}{\Stop}$, and its expression evaluation part, with the expression evaluation and conditional jump of the \texttt{RISC} program that $\Whileg{e}{c}$ was compiled to by $\CompileCmd$.

        It relies on $c$ being related by $\RefRelWR$ to its compiled \texttt{RISC} counterpart when started with initialisation states judged valid by $\CompiledCmdConfigConsistent$.

    \item \texttt{while\_inner}: This case relates some program $\Seqg{c_I}{\Whileg{e}{c}}$ reachable from $\Seqg{c}{\Whileg{e}{c}}$ to the loop body part of the \texttt{RISC} program compiled from $\Whileg{e}{c}$.

        It relies on $c_I$ being related by $\RefRelWR$ to its compiled \texttt{RISC} counterpart at the appropriate program counter offset.

        It also carries around the same reliance on $c$ being related by $\RefRelWR$ to its compiled \texttt{RISC} counterpart for all initialisation states judged valid by $\CompiledCmdConfigConsistent$.

    \item \texttt{while\_loop}: This case handles epilogue steps for the inner loop body program, and the final jump back to the beginning of the While-loop.

        It requires $\RefRelWR$ to relate the terminated \texttt{While} program to the end of the compiled loop body, and furthermore also carries around the same reliance on $c$ being related by $\RefRelWR$ to its compiled \texttt{RISC} counterpart for all initialisation states judged valid by $\CompiledCmdConfigConsistent$.

\end{itemize}

\section{Code listing for the case of $\CompileCmd$ for if-conditionals}\label{sec:compile-cmd-if}

This code listing has been adapted slightly to improve the clarity of the presentation.
$\Phi \sqcap_R \Phi'$ denotes the subset of mappings on which $\Phi$ and $\Phi'$ agree.

\lstset{mathescape}
\begin{lstlisting} [caption={Implementation of $\CompileCmd$ case for $\ITEg{e}{c_1}{c_2}$}]
compile_cmd C l nl (If e c$_1$ c$_2$) =
  (let (P$_e$, r, C$_1$, fail$_e$) = (compile_expr C {} l e);
       (br, nl') = (nl, Suc nl); (ex, nl'') = (nl', Suc nl');
       (P$_1$, l$_1$, nl$_1$, C$_2$, fail$_1$) = (compile_cmd C$_1$ None nl'' c$_1$);
       (P$_2$, l$_2$, nl$_2$, C$_3$, fail$_2$) = (compile_cmd C$_1$ (Some br) nl$_1$ c$_2$);
       (* Pre-compilation check ensures asmrec C$_2$ = asmrec C$_3$ *)
       C' = (regrec C$_2$ $\sqcap_R$ regrec C$_3$, asmrec C$_2$)
    in (P$_e$ @ [((if P$_e$ = [] then l else None, Jz br r), C$_1$)] @
        P$_1$ @ [((l$_1$, Jmp ex), C$_2$)] @ P$_2$ @ [((l$_2$, Nop'), C$_3$)],
        Some ex, nl$_2$, C', fail$_e$ $\lor$ fail$_1$ $\lor$ fail$_2$))
\end{lstlisting}

\section{Proof sketch for decomposition principle soundness result}\label{sec:elided-decomp}

\begin{theorem} [Soundness of $\SecureRefineSimpler$] \label{thm:secure-refinement-simpler-sound-elided}
\begin{align*}
\SecureRefineSimpler\ \mathcal{B}\ \mathcal{R}\ \mathcal{I}\ \abssteps\ \implies \SecureRefinement\ \mathcal{B}\ \mathcal{R}\ \mathcal{I}
\end{align*}
\end{theorem}
\begin{proof}
    The only obligation for $\SecureRefinement$ (\autoref{def:secure-refinement}) not obtained immediately from $\SecureRefineSimpler$ (\autoref{def:secure-refinement-simpler}) is the cube-shaped $\CouplingInvPres$ (\autoref{fig:coupling-inv-pres}).

    The front face of the cube is just ordinary square-shaped refinement preservation (depicted in \autoref{fig:decomp-R}), given to us by $\SecureRefineSimpler$.
    This gives us that a single concrete step from $\lc_{1C}$ is simulated by $n$ abstract steps $\lc_{1A}$, where $n$ is given by $\abssteps$.

    We are then obliged to prove a simulation in the other direction (the back face of the cube), that $n$ abstract steps from all configurations $\lc_{2A}$ related by $\mathcal{B}$ to $\lc_{1A}$ are simulated by some concrete step from $\lc_{2C}$ related by $\mathcal{R}$ to $\lc_{2A}$ and by $\mathcal{I}$ to $\lc_{1C}$.

    Here, we lean on the determinism of the abstract program's evaluation semantics (required by the theory) to flip the direction of simulation, knowing that $n$ abstract steps from $\lc_{2A}$, simulating a single concrete step from $\lc_{2C}$, could only be the very same $n$ abstract steps from $\lc_{2A}$ that we were required to consider.
    This allows us to use once again the square-shaped refinement preservation (\autoref{fig:decomp-R}) given to us by $\SecureRefineSimpler$.

    Consistency of refinement pacing and stopping behaviour (depicted in \autoref{fig:decomp-abs-steps}) given by $\SimplerRefinementSafe$ (\autoref{def:simpler-refinement-safe}) then respectively ensure that $n$ (via $\abssteps$) is the correct number of abstract steps to consider, and that there will indeed be a concrete step from $\lc_{2C}$ to drive the matching simulation step.

    Finally, the remainder of $\SimplerRefinementSafe$ (depicted in \autoref{fig:decomp-I}) discharges the requirement of closedness and modes-equality maintenance of $\mathcal{I}$ under lockstep execution, demanded by the bottom face of the cube.
\end{proof}

\section{More details on the {\mdseries\Covern} \WRCompiler}\label{sec:elided-compiler-section}

\subsection{Register allocation scheme model} \label{sec:reg-alloc}

We model the (user-supplied) register allocation scheme by two functions $reg\_alloc$ and $reg\_alloc\_cached$ on the \emph{register record} $\Phi$ (see \autoref{sec:comprec}) and the set $A$ of registers whose contents are needed to evaluate the current expression.
In order to avoid loading from memory unnecessarily, the compiler may first call $reg\_alloc\_cached\ \Phi\ A\ v$ to identify a register that $\Phi$ records as already containing the variable $v$.
When the compiler needs a fresh register, it will call $reg\_alloc\ \Phi\ A$.
Neither function is allowed to allocate a register in $A$, so the allocator is permitted to fail if it cannot find any suitable register.
As mentioned in \autoref{sec:target-lang} we assume there are enough registers to service the expressions in the source program.
Also, registers typically become available again as expression evaluation is resolved.

\subsection{Label allocation and sequential composability}\label{sec:labels-seq}

For allocating natural numbers to use as labels for \texttt{RISC} instructions the \WRCompiler ensures freshness merely by using the highest number reached so far on a ``next label'' counter ($\nl$ in the invocation example~(\ref{eqn:example-invocation})), incrementing the counter before passing it along to subsequent calls, and outputting the next available unused label on return (as $\nl'$ in the example).

We define two \texttt{RISC} programs $P_1, P_2$ to be $\Joinable$ if they are both:
\begin{itemize}
    \item $\JoinableFwd$: $P_1$ only ever jumps to labels that are either
        \begin{itemize}
            \item labelling an instruction in $P_1$ itself, or
            \item the label of the very first instruction in $P_2$.
        \end{itemize}
    \item $\JoinableBwd$: $P_2$ does not jump to any of the labels of instructions in $P_1$.
\end{itemize}
We prove a lemma that says that two \texttt{RISC} programs that were compiled by the \WRCompiler \emph{consecutively}---in the sense that the relevant outputs from the first call are fed directly into the second call---are $\Joinable$.

\subsection{More detail on $\CompilerInputReqs$ and the \WRCompiler proofs}\label{sec:elided-compiler-proofs}

The first two requirements to the predicate $\CompilerInputReqs\ C\ l\ \nl\ c$ were given in \autoref{sec:comprec}.
Its other two requirements reflect that the terminated \texttt{While} program $\Stop$ has no valid compilation, and that the initial label (if provided) must be valid (see \autoref{sec:labels-seq} for more information on label allocation).

\begin{definition} [Requirements on inputs to $\CompileCmd$]
\[
\begin{aligned}
    \CompilerInputReqs\ C\ l\ \nl\ c\ \equiv \ \ 
    & c \ne \Stop\ \land \ 
      (\forall x. \ l = \Some\ x \longrightarrow x < \nl)\ \land \\
    & \NoUnstableExprs\ c\ C\ \land \ 
      \RegrecStable\ C
\end{aligned}
\]
\end{definition}

These input conditions give us enough information to prove that every instruction of a $CompRec$-annotated \texttt{RISC} program output by a successful run of $\CompileCmd$ is annotated by a stable register record, and that the output $CompRec$'s register record is also stable:

\begin{lemma} [Successful compilations output only stable register records] \label{thm:compiled-cmd-regrec-stable}
\[
  \mprset{vskip=0.5ex}
  \inferrule{
    (\PCs, l', \nl', C', \False) = \CompileCmd\ C\ l\ \nl\ c \and
      \CompilerInputReqs\ C\ l\ \nl\ c
  } {
      (\forall \pc < \length\ \PCs.\ \RegrecStable\ (\mapsnd\ \PCs\ !\ \pc))\,\land\,
      \RegrecStable\ C'
  }
\]
\end{lemma}
\begin{proof}
    By induction on the structure of the \texttt{While} language program $c$, making reference to the implementation of $\CompileCmd$.
    For cases that must compile expressions, we furthermore prove and make use of a lemma by induction on the structure of expressions, making reference to the implementation of the expression compiler function $\CompileExpr$ called by $\CompileCmd$.
    In essence, we prove that (sub)expressions that appear in register records must be stable, for two reasons.
    Firstly, they are always only ever subexpressions over variables that must have been stable in the input program when their contents were first loaded into registers.
    Furthermore, when compiling an $\LockRel$, the \WRCompiler will always flush all register records that make reference to any variables that the $\LockRel$ makes unstable.
\end{proof}

Before proceeding, we name the parts of $\CompiledCmdConfigConsistent$ more explicitly:

\begin{definition} [Configuration consistency requirements for compiled commands] \label{def:compiled-cmd-config-consistent-elided}
\begin{align*}
    & \CompiledCmdConfigConsistent\ C\ \regs\ \mds\ \mem\ \equiv \\
    & \RegrecMemConsistent\ (\Regrec\ C)\ \regs\ \mem\ \land \ 
      \AsmrecMdsConsistent\ (\Asmrec\ C)\ \mds
\end{align*}
\end{definition}

\begin{definition} [Consistency between a register record, register bank, and shared memory] \label{def:regrec-mem-consistent-elided}
\[
    \RegrecMemConsistent\ \Phi\ \regs\ \mem\ \equiv \ 
    \forall r\ e.\ \Phi\ r = \Some\ e \longrightarrow regs\ r = \mathsf{ev_{exp}}\ \mem\ e
\]
\end{definition}

\begin{definition} [Consistency between an assumption record and a mode state] \label{def:asmrec-mds-consistent-elided}
 \[
     \AsmrecMdsConsistent\ \mathcal{S}\ \mds\ \equiv \ 
     \mathcal{S} = (\mds\ \mathbf{AsmNoW},\ \mds\ \mathbf{AsmNoRW})
 \]
\end{definition}

\begin{lemma} [$\RefRelWR$ preserves modes and memory] \label{thm:preserves-modes-mem-R_wr-elided}
  $\PreservesMM\ \RefRelWR$
\end{lemma}
\begin{proof}
By induction on the structure of $\RefRelWR$.
For all cases of $(\lc_{w}, \lc_{r}) \in \RefRelWR$, $\LCSameMdsMem{\lc_w}{\lc_r}$ is either asserted directly by the guards or obtainable from the inductive hypothesis.
\end{proof}

\begin{lemma} [$\RefRelWR$ is closed under changes by others] \label{thm:closed-others-R_wr-elided}
  $\ClosedOthers\ \RefRelWR$
\end{lemma}
\begin{proof}
    By induction on the structure of $\RefRelWR$.
    Changes by others (\autoref{def:closed-others}) only modify $\Writable$ variables the same way for both configurations, so preservation of $\LCSameMdsMemOp$ is immediate.
    Also, $\RegrecMemConsistent$ is unaffected because $\CompileCmd$ only creates $\RegrecStable$ records (referring to no $\Writable$ variables).
    No other $\RefRelWR$ guards mention shared memory.
\end{proof}

\begin{lemma} [Successfully compiled programs maintain config consistency requirements]\label{thm:compiled-cmd-eval-maintains-comprec-consistency-elided}
\[
  \mprset{vskip=0.5ex}
  \inferrule {
    (\PCs, l', \nl', C', \failed) = \CompileCmd\ C\ l\ \nl\ c \and
      \CompilerInputReqs\ C\ l\ \nl\ c \and \\
    \failed = \False \and
      \pc < \length\ \PCs \and
      P = \mapfst\ \PCs \and
      \Cs = \mapsnd\ \PCs \and \\
    \CompiledCmdConfigConsistent\ (\Cs\ !\ \pc)\ \regs\ \mds\ \mem \and \\
    \LocalConfRISC{\pc}{P}{\regs}{\mds}{\mem}
      \EvalStepRISC
      \LocalConfRISC{\pc'}{P}{\regs'}{\mds'}{\mem'})
  } {
    \CompiledCmdConfigConsistent\ 
      (\mathtt{if}\ \pc' < \length\:P\ 
       \mathtt{then}\ (\Cs\:!\:\pc')\ 
       \mathtt{else}\ C')\:
      \regs'\:\mds'\:\mem'
  }
\]
\end{lemma}
\begin{proof}
    We in fact prove it separately for $\RegrecMemConsistent$ and $\AsmrecMdsConsistent$, in both cases by induction on the structure of the \texttt{While} program $c$.
    In each case, we use the simplifiers for the $\CompileCmd$ implementation to yield the corresponding \texttt{RISC} program fragment in question, and then prove the lemma for each of the possible locations of $\pc$ in the compiled program.
    For both proofs, there is some trickiness in accounting for (and ruling out) which destination $\pc'$ must be considered for each of these cases of $\pc$, particularly for those \texttt{While} programs that compile to \texttt{RISC} programs that may have jumps in them.

    Control flow trickiness aside, the intuition for $\RegrecMemConsistent$ is that it tests the correctness of the compilation of expressions, and so for this we must prove a sub-lemma for maintenance of $\CompiledCmdConfigConsistent$ by induction on the structure of expressions $e$ that are encountered in the \texttt{While} programs $\ITEg{e}{c_1}{c_2},\ \Whileg{e}{c'},\ \Assign{v}{e}$.
    Additionally, $\LockRel{}$ flushes register record entries mentioning variables that are to become unstable, and $\Whileg{e}{c'}$ conservatively flushes entries to force evaluation of the loop condition expression.
    This is safe trivially because flushing entries can never make a consistent register record inconsistent.
    The rest of the cases for $c$ are straightforward because they do not touch the register record.

    Then for $\AsmrecMdsConsistent$, the substantial part of the proof is as a test of the correctness of the compiler's bookkeeping of assumptions being consistent with the semantics of $\LockAcq{}$ and $\LockRel{}$.
    The other cases for $c$ do not touch the mode state.
\end{proof}

\begin{lemma} [Correctness of the expression compiler] \label{thm:compile-expr-correct-elided}
\[
    (\PCs, r, C', \False) = \CompileExpr\ C\ A\ l\ e\ 
    \implies (\Regrec\ C')\ r = \Some\ e
\]
\iffalse
\[
    \mprset{vskip=0.5ex}
    \inferrule {
        (\PCs, r, C', \failed) = \CompileExpr\ C\ A\ l\ e \and
        \failed = \False
    } {
        (\Regrec\ C')\ r = \Some\ e
    }
\]
\fi
\end{lemma}
\begin{proof}
    By induction on the structure of expressions $e$, using the simplification rules for the implementation of $\CompileExpr$, and also relying on assumptions of correctness of the register allocation scheme supplied by the instantiator of the theory.
\end{proof}

\begin{lemma} [$\RefRelWR$ is a refinement paced by $\AbsStepsWR$] \label{thm:R_wr-refinement-abs-steps_wr-elided}
\begin{align*}
    & \forall \lc_w\ \lc_r.\ (\lc_w, \lc_r) \in \RefRelWR \longrightarrow \ 
      (\forall \lc_r'.\ \lc_r \EvalStepRISC \lc_r' \longrightarrow \\
    & \qquad (\exists \lc_w'.\ \lc_w \NEvalStepWhile{(\AbsStepsWR\ \lc_w\ \lc_r)} \lc_w'\ \land \ (\lc_w', \lc_r') \in \RefRelWR))
\end{align*}
\end{lemma}
\begin{proof}
    By induction on the structure of $\RefRelWR$.

    The base case \texttt{stop} is immediate, because it pertains to a terminated \texttt{While} and \texttt{RISC} program.
    The base cases that proceed in one step to a terminating program configuration (\texttt{skip\_nop}, \texttt{assign\_store}, \texttt{lock\_acq}, \texttt{lock\_rel}) are fairly straightforward because after dealing with the single step, the resulting obligation can then be handled by the \texttt{stop} case.
    This leaves the last remaining base case \texttt{assign\_expr}, which proceeds in one step either to itself, or to \texttt{assign\_store}.
    In all of these cases, we use \autoref{thm:compiled-cmd-eval-maintains-comprec-consistency-elided} to obtain the preservation of the guards demanded by the $\RefRelWR$ introduction rule for the destination configuration of the step.
    Particularly, the \texttt{assign\_store} case must make use of $\RegrecMemConsistent$ and the correctness of $\CompileExpr$ (\autoref{thm:compile-expr-correct-elided}) in order to ensure that once the expression evaluation result is written back to shared memory, $\LCSameMdsMem{\lc_w'}{\lc_r'}$ holds as demanded by the \texttt{stop} case.

    The inductive cases that concern expression evaluation (\texttt{if\_expr}, \texttt{while\_expr}) are much like \texttt{assign\_expr} in that they have the possibility of progressing in one step to themselves.
    Unlike \texttt{assign\_expr} however, their other possibility is a conditional jump based on the result of that expression.
    Again we use \autoref{thm:compile-expr-correct-elided} to obtain that the result is an accurate calculation of the expression, and this time we prove by the two different cases whether \texttt{if\_expr} ends up in \texttt{if\_c1} or \texttt{if\_c2}, or if \texttt{while\_expr} ends up in \texttt{while\_inner} or at \texttt{stop} (having jumped to the exit label).
    In these cases, the guards over which the inductive references to $\RefRelWR$ have been quantified are versatile enough to discharge themselves (when \texttt{*\_expr} steps to itself), or to discharge any reachable initial starting state for the nested compiled \texttt{RISC} program, given that \autoref{thm:compiled-cmd-eval-maintains-comprec-consistency-elided} ensures the invariance of these guards.

    This just leaves the inductive cases that pertain to configurations inside a nested compiled \texttt{RISC} program (\texttt{if\_c1}, \texttt{if\_c2}, \texttt{while\_inner}), or at the end of one (\texttt{epilogue\_step}, \texttt{while\_loop}).
    In these cases, the inductive hypotheses obtained from the inductive reference to $\RefRelWR$ are always enough to satisfy the guards demanded by the possible destination cases.
    Like in the proof of \autoref{thm:compiled-cmd-eval-maintains-comprec-consistency-elided}, the trickiness mostly comes from accounting for all the possible cases of control flow (ruling out spurious destinations) that need to be considered.
\end{proof}

\begin{lemma} \label{thm:simpler-refinement-safe_wr-elided}
$\begin{aligned}
\inferrule{
    \StrongLowBisimMM\ \mathcal{B} \and
    \NoHighBranching\ \mathcal{B}
} {
    \SimplerRefinementSafe\ \mathcal{B}\ \RefRelWR\ \CouplInvWR\ \AbsStepsWR
}\end{aligned}$
\end{lemma}
\begin{proof}
    \autoref{def:simpler-refinement-safe} gives us the following obligations.

    For consistent stopping behaviour, we prove a lemma that \texttt{RISC} programs stop if and only if their $\pc$ is outside the program text $P$, i.e. $\pc > \length\ P$.
    Because $\CouplInvWR$ equates $\pc$ and $P$ for the two configurations, then clearly both have identical stopping behaviour.

    For consistency of change in timing behaviour, $\AbsStepsWR$ depends only on \texttt{While} and \texttt{RISC} program locations, and $\NoHighBranching$ and $\CouplInvWR$ forces them (resp.) to be equal for the local configurations under consideration.

    For closedness of $\CouplInvWR$ under lockstep execution, the only non-straightforward cases to consider are conditional branching, and the locking primitives.
    For conditional branching, we use $\NoHighBranching$ for $\mathcal{B}$ with memory preservation via $\RefRelWR$ (\autoref{thm:preserves-modes-mem-R_wr}) to ensure that the conditional branching outcome is the same on both sides.

    Finally, as the only operations that touch mode state, the locking primitives are the only non-straightforward cases for mode state equality maintenance under lockstep execution.
    As all lock memory is classified $\Low$ (see \autoref{sec:dma-lock-interp-reqs}),
    we use $\StrongLowBisimMM$ for $\mathcal{B}$ with memory preservation via $\RefRelWR$ to ensure the \texttt{RISC} configurations behave consistently.
\end{proof}

\begin{lemma}
$\begin{aligned}
\inferrule{
    \StrongLowBisimMM\ \mathcal{B} \and
    \NoHighBranching\ \mathcal{B}
} {
    \SecureRefineSimpler\ \mathcal{B}\ \RefRelWR\ \CouplInvWR\ \AbsStepsWR
}\end{aligned}$
\end{lemma}
\begin{proof}
    Referring to \autoref{def:secure-refinement-simpler}, the obligations pertaining only to $\RefRelWR$ and $\AbsStepsWR$ are discharged by \autoref{thm:R_wr-refinement-abs-steps_wr}, \autoref{thm:closed-others-R_wr}, and \autoref{thm:preserves-modes-mem-R_wr}.
    Pertaining to $\CouplInvWR$: clearly $\CouplInvWR$ is symmetric, and furthermore it is $\CgConsistent$ (\autoref{def:cg-consistent}) because the actions over which $\CouplInvWR$ must be closed modify only the shared memory, and $\CouplInvWR$ places only restrictions on the program text and current location.
    The final obligation is discharged by \autoref{thm:simpler-refinement-safe_wr-elided}.
\end{proof}

\begin{theorem} [Successful compilations are refinements in $\RefRelWR$] \label{thm:compile-cmd_correctness_R_wr-elided}
\[
\mprset{vskip=0.5ex}
\inferrule{
    (\PCs, l', \nl', C', \failed) = \CompileCmd\ C\ l\ \nl\ c \and
      \CompilerInputReqs\ C\ l\ \nl\ c \and \\
    \failed = \False \and
      \CompiledCmdConfigConsistent\ C\ \regs\ \mds\ \mem \and
      P = \mapfst\ \PCs
} {
    (\LocalConfWhile{c}{\mds}{\mem},
     \LocalConfRISC{0}{P}{\regs}{\mds}{\mem}) \in \RefRelWR
}
\]
\end{theorem}
\begin{proof}
    By induction on the structure of \texttt{While}.
    The compiler input and initial configuration conditions we impose allow us to have each of
    $\Skip$, $\Seqg{\cmd}{\cmd}$, $\ITEg{exp}{\cmd}{\cmd}$,
    $\Whileg{exp}{\cmd}$, $\Assign{v}{exp}$,
    $\LockAcq{k}$, and $\LockRel{k}$
    and their compiled output meet the guards of the introduction rules for the cases
    \texttt{skip}, \texttt{seq}, \texttt{if\_expr}, \texttt{while\_expr}, \texttt{assign\_expr}, \texttt{lock\_acq}, and \texttt{lock\_rel} of $\RefRelWR$ that were designed for them respectively.
\end{proof}
\fi % ifPreprint

\end{document}